\documentclass[a4paper,onecolumn,11pt,accepted=2022-12-19]{quantumarticle}
\pdfoutput=1
\usepackage[T1]{fontenc}
\usepackage[english]{babel}
\usepackage[utf8]{inputenc}
\usepackage{amsmath}
\usepackage{amsthm}
\usepackage{amsfonts}
\usepackage{geometry}
\usepackage{graphicx}
\usepackage{physics}
\usepackage{amssymb}
\usepackage{ulem}
\usepackage[numbers]{natbib}
\usepackage{hyperref}
\usepackage[toc,page]{appendix}
\usepackage{mathtools}
\usepackage{authblk}

\newtheorem{theorem}{Theorem}
\newtheorem*{theorem*}{Theorem}
\newtheorem{definition}{Definition}
\newtheorem{notation}{Notation}

\title{Multipartite Intrinsic Non-Locality and Device-Independent Conference Key Agreement}

\author[1,5]{Aby Philip} 
\author[2,4]{Eneet Kaur} 
\author[3]{Peter Bierhorst} 
\author[1,6]{Mark M. Wilde} 
\affil[1]{Hearne Institute for Theoretical Physics, Department of Physics and Astronomy, and Center for Computation and Technology, Louisiana State University, Baton Rouge, Louisiana 70803, USA}  
\affil[2]{Institute for Quantum Computing and Department of Physics and Astronomy, University of Waterloo, Waterloo, Ontario N2L 3G1, Canada}
\affil[3]{Department of Mathematics, University of New Orleans, Louisiana 70148, USA} 
\affil[4]{Wyant College of Optical Sciences, University of Arizona, Tucson, Arizona 85721, USA} 
\affil[5]{School of Applied and Engineering Physics, Cornell University, Ithaca, New York 14850, USA}
\affil[6]{School of Electrical and Computer Engineering, Cornell University, Ithaca, New York 14850, USA}

\allowdisplaybreaks

\begin{document}

\maketitle

\begin{abstract}
    In this work, we introduce multipartite intrinsic non-locality as a method for quantifying resources in the multipartite scenario of device-independent (DI) conference key agreement. We prove that multipartite intrinsic non-locality is additive, convex, and monotone under a class of free operations called local operations and common randomness. As one of our technical contributions, we establish a chain rule for two variants of multipartite mutual information, which we then use  to prove that  multipartite intrinsic non-locality is additive. This chain rule may be of independent interest in other contexts. All of these properties of multipartite intrinsic non-locality are helpful in establishing  the main result of our paper:  multipartite intrinsic non-locality is an upper bound on secret key rate in the general multipartite scenario of DI conference key agreement. We discuss various examples of DI conference key protocols and compare our upper bounds for these protocols with known lower bounds. Finally, we calculate upper bounds on recent experimental realizations of DI quantum key distribution.
\end{abstract}  

\tableofcontents

\section{Introduction}
    In principle, quantum key distribution (QKD) can produce a secret key secured by the laws of physics \cite{ch1984quantum,PhysRevLett.67.661,10.1145/382780.382781}. In the device-dependent setting of QKD, it is assumed that the devices possessed by Alice and Bob are perfectly characterized and trusted; i.e., the measurements applied and the states used are assumed to be known and certified. However, after several experiments implementing QKD protocols, researchers have found this assumption to be too restrictive.

    To combat our reliance on some of these strong assumptions underpinning QKD, several scenarios have been developed with varying degrees of trust in the measurements and states used. In QKD, if the measurements, states, or devices possessed by one of the parties are not trusted, the scenario is called one-sided device-independent QKD \cite{PhysRevLett.106.110506,PhysRevA.85.010301}. If all devices involved are deemed to be untrustworthy, the scenario is called device-independent QKD~\cite{743501,PhysRevLett.98.230501,arnon2018practical,PhysRevLett.113.140501}. 

    Researchers have established upper bounds on the secret key agreement capacity for all the scenarios described above \cite{takeoka2014fundamental,kaur2020fundamental} (see also \cite{winczewski2022limitations}). The basic idea behind these upper bounds  comes from a classical information measure called intrinsic information \cite{748999}. Intrinsic information inspired the squashed-entanglement upper bound for device-dependent QKD \cite{christandl2004squashed,takeoka2014fundamental}, and  squashed entanglement in turn inspired  the development of quantum intrinsic non-locality \cite{kaur2020fundamental} and quantum intrinsic steerability \cite{KWW17}. These latter quantities serve as upper bounds for device-independent QKD and one-sided device-independent QKD, respectively, as shown in \cite{kaur2020fundamental}. Along with being upper bounds on a certain cryptographic task, these quantities are also resource quantifiers for Bell non-locality and steerability, respectively.

    Here, we go beyond device-independent QKD  and Bell non-locality for two parties and address device-independent (DI) conference key agreement \cite{ribeiro2018fully,murta2020quantum} and multipartite non-locality. Conference key agreement is the task of distributing secret key among more than two users, as encountered in the context of quantum networks. Part of the interest in this task comes from the fact that a protocol based on genuinely multipartite entangled states can achieve higher rates of conference key agreement than a protocol based on a combination of bipartite entangled states \cite{epping2016large}. Just as Bell non-locality is the key resource for DIQKD, one would expect multipartite non-locality to be the key resource in DI conference key agreement.
    
    Here, we propose a resource quantifier for multipartite non-locality called multipartite intrinsic non-locality. We base instances of this resource quantifier on  total correlation and dual total correlation \cite{watanabe1960information} (see also \cite{YHHHOS09,AHS08}), which generalize mutual information to the multipartite case. Total correlation and dual total correlation have previously been used to establish upper bounds on entanglement distillation and secret key agreement capacities of quantum broadcast channels \cite{7438836}; see  \cite{YHHHOS09} for its use in establishing an upper bound on distillable secret key and distillable entanglement of a multipartite state. We use multiparite intrinsic non-locality to derive upper bounds on the ultimate rate at which device-independent (DI) conference key agreement is possible.
    
    To show that our quantity is indeed a useful upper bound, it is necessary to prove that it is additive. In order to prove additivity (and other useful properties) of multipartite intrinsic non-localities, we establish a chain rule for total correlation and dual total correlation of two rounds of the conference key agreement protocol in Section~\ref{sec:mul-int-non}. The chain rule for total correlation expresses the total correlation of two rounds of the conference key agreement protocol as the sum of total correlation terms related to the individual rounds of the conference key agreement protocol and other information theoretic quantities. These additional information-theoretic quantities are expressed in terms of conditional mutual information. 
    For our paper, we derive a chain rule for total correlation and dual total correlation that meets the aforementioned criteria and holds for all finite $M$. Such a broadly applicable chain rule is not obtained in \cite{kaur2020fundamental}. 
    
    In what follows, we first discuss no-signaling and quantum correlation and then proceed to no-signaling and quantum extensions. After that, we define a quantum tripartite intrinsic non-locality, which is based on tripartite total correlation, and prove that it is indeed additive, convex, and monotone under local operations and common randomness. We then define the multipartite intrinsic non-localities using total correlation and dual total correlation, starting by defining and discussing multipartite intrinsic non-locality based on total correlation and then moving on to the one defined in terms of dual total correlation.  We establish important identities (our chain rule) for total correlation and dual total correlation that allow us to use arguments similar to those presented for the tripartite scenario to prove that the multipartite intrinsic non-localities, presented in this paper, are additive, convex upper bounds on device-independent conference key agreement capacity in the general $M$-partite case. Then, we give a general overview of device-independent conference key agreement for the tripartite case and define the DI conference key agreement capacity. Finally, we show that tripartite intrinsic no-locality is an upper bound on DI conference key agreement capacity for the tripartite situation and provide arguments to show that multipartite intrinsic non-locality upper bounds the $M$-partite DI conference key agreement capacity, for all finite $M$.

    As other contributions, we calculate upper bounds on both quantum tripartite intrinsic non-localities using eavesdropper attacks similar to those from \cite{kaur2020fundamental} and \cite{arnon2021upper}, which were used to calculate upper bounds on quantum intrinsic non-locality. We plot quantum tripartite intrinsic non-locality versus parity-CHSH violation under these attacks, and we compare these to previously calculated lower bounds from \cite{ribeiro2018fully}. We also consider a noise model in which each share of the tripartite state passes through a qubit depolarizing channel. We plot quantum tripartite intrinsic non-localities versus the depolarizing parameter~$p_{\text{dep}}$ for this noise model and compare them to the lower bound from~\cite{ribeiro2018fully}.

    The rest of this paper is structured as follows.  
    Section~\ref{sec:no-sig} discusses no-signaling constraints, no-signaling extensions, and quantum extensions, focusing especially on the tripartite case. Section~\ref{sec:tri-int-non} contains the definition of tripartite intrinsic non-locality and proves that it is additive using a chain rule, which we derive here.  Sections~\ref{sec:mul-int-non} and~\ref{sec:dual-total} generalize tripartite intrinsic non-locality and all of its properties to the multipartite case using total correlation and dual total correlation, respectively, and generalizations of the aforementioned chain rule. Section~\ref{sec:di-cka-capacity} introduces a general form of a DI conference key agreement protocol and its associated capacity. Then, we show that tripartite intrinsic non-locality is an upper bound on the tripartite device-independent conference key agreement capacity. Section~\ref{sec:ex-ex} contains some examples of our upper bound calculated under various attacks by an eavesdropper. Additionally, in Section~\ref{sec:UpBo-DIQKD}, we evaluate upper bounds for recent experimental protocols implementing device-independent quantum key distribution (DIQKD) \cite{zhang2021experimental,schwonnek2021device,liu2021highspeed}.  Section~\ref{sec:con-clues} contains our conclusions and possible directions for future work.

\section{Correlations, No-Signaling Conditions, and Quantum Extensions}\label{sec:no-sig}

    First, let us define the types of correlations that we are concerned with in this paper: no-signaling correlations and quantum correlations. Let us begin by discussing no-signaling correlations.

    No-signaling conditions impose constraints on correlations, which imply that parties sharing the correlation cannot use it alone to communicate; i.e., no party can infer the input choices of another party based solely on their own outputs \cite{beckman2001causal}. On a technical level, no-signaling conditions imply that tracing over subsets of outputs of a correlation results in tracing over the corresponding inputs \cite{brunner2014bell}. These conditions are relevant in our scenario as it is necessary to verify that the correlations observed are from the state and measurement choices shared by the participants and not from classical communication when the input choices are made. Compliance with no-signaling conditions can be enforced by imposing space-like separation between measuring parties or constructing other barriers to prevent communication.
    
    No-signaling conditions for the tripartite scenario are as follows:
    \begin{align}\label{nosig}
        \sum_{a}p(a,b,c\vert x,y,z) &= \sum_{a}p(a,b,c\vert \bar{x},y,z) = p(b,c\vert y,z) \quad \forall x,\bar{x},\notag\\
        \sum_{b}p(a,b,c\vert x,y,z) &= \sum_{b}p(a,b,c\vert x,\bar{y},z) = p(a,c\vert x,z) \quad \forall y,\bar{y},\notag\\
        \sum_{c}p(a,b,c\vert x,y,z) &= \sum_{c}p(a,b,c\vert x,y,\bar{z}) = p(a,b\vert x,y) \quad \forall z,\bar{z}.
    \end{align}
    The set of all correlations that satisfy the above three conditions in \eqref{nosig} are called no-signaling correlations. 
    The no-signaling conditions above can also equivalently be expressed in terms of conditional mutual information as follows:
    \begin{align}
    \label{eq:nosig-info-th-equiv}
        I(X;BC\vert YZ)_{\rho} = I(Y;AC\vert XZ)_{\rho} = I(Z;AB \vert XY)_{\rho} =0,
    \end{align}
    where
    \begin{align}
        \rho_{ABCXYZ} = \sum_{a,b,c,x,y,z}q(x,y,z)p(a,b,c\vert x,y,z)\ketbra*{abcxyz}{abcxyz}_{ABCXYZ},
    \end{align}
    $p(a,b,c\vert x,y,z)$ is a no-signaling correlation, and the conditional mutual information of random variables $K$, $L$, and $M$ is defined as
    \begin{equation}
    I(K;L\vert M) \coloneqq  H(KM) + H(LM) - H(M) - H(KLM),   
    \end{equation}
     where $H$ denotes the entropy. It suffices to take the input distribution $q$ to be uniform. Note that the conditions in \eqref{nosig} imply the following ones, by tracing over two of the outputs, rather than just one:
    \begin{align}
    \label{nosig2}
        I(YZ;A\vert X)_{\rho} = I(XZ;B\vert Y)_{\rho} = I(XY;C \vert Z)_{\rho} =0.
    \end{align}

    Now we move on to quantum correlations. Consider the following scenario: Alice, Bob and Charlie are given a share of a tripartite quantum state $\rho_{\tilde{A}\tilde{B}\tilde{C}}$ that is distributed to them by a possibly unknown entity, and each party has access to a black box with which they can interact classically. For each classical input, the corresponding black box applies a positive operator-valued measure (POVM) on its respective share of the tripartite state. After the application of the POVM, the box outputs a classical value that is recorded by the corresponding participant. The correlation that is obtained using the aforementioned process is of the following form:
        \begin{align}
            p(a,b,c\vert x,y,z) = \operatorname{Tr}\!\left([\Pi^{(x)}_{a}\otimes\Pi^{(y)}_{b}\otimes\Pi^{(z)}_{c}]\rho_{\tilde{A}\tilde{B}\tilde{C}}\right),
            \label{eq:q-corr-trip}
        \end{align}
    where $\{\Pi^{(x)}_{a}\}_{a}$, $\{\Pi^{(y)}_{b}\}_{b}$, and $\{\Pi^{(z)}_{c}\}_{c}$ are POVMs. Correlations of the form described in~\eqref{eq:q-corr-trip} are called quantum correlations. Quantum correlations are a subset of no-signaling correlations. This fact can easily be seen in the example analysis below:
    \begin{align}
        \sum_{a}p(a,b,c\vert x,y,z) &= \sum_{a}\operatorname{Tr}\!\left([\Pi^{(x)}_{a}\otimes\Pi^{(y)}_{b}\otimes\Pi^{(z)}_{c}]\rho_{\tilde{A}\tilde{B}\tilde{C}}\right) \\
        &= \operatorname{Tr}\!\left([\mathbb{I} \otimes \Pi^{(y)}_{b}\otimes\Pi^{(z)}_{c}]\rho_{\tilde{A}\tilde{B}\tilde{C}}\right)\\
        &= \operatorname{Tr}\!\left([ \Pi^{(y)}_{b}\otimes\Pi^{(z)}_{c}]\rho_{\tilde{B}\tilde{C}}\right)\\
        & =p(b,c\vert y,z).
    \end{align}

    Since we are looking at non-locality for the sake of a cryptographic task, it is necessary that we delineate the power that the eavesdropper possesses. We do so by allowing the eavesdropper to possess either a no-signaling extension or a quantum extension.
    No-signaling extensions are extensions of a correlation that obey the above no-signaling constraints and can be expressed as follows:
    \begin{align}\label{nosig_ext}
        \sum_{a}p(a,b,c\vert x,y,z)\rho_{E}^{abcxyz} &= \sum_{a}p(a,b,c\vert \bar{x},y,z)\rho_{E}^{a,b,c,\bar{x},y,z}  \quad \forall x,\bar{x},\nonumber\\
        \sum_{b}p(a,b,c\vert x,y,z)\rho_{E}^{abcxyz} &= \sum_{b}p(a,b,c\vert \bar{y},x,z)\rho_{E}^{a,b,c,\bar{y},x,z}  \quad \forall y,\bar{y},\nonumber\\
        \sum_{c}p(a,b,c\vert x,y,z)\rho_{E}^{abcxyz} &= \sum_{c}p(a,b,c\vert \bar{z},y,x)\rho_{E}^{a,b,c,\bar{z},y,x}  \quad \forall z,\bar{z}.
    \end{align}

    A type of no-signaling extensions, in which we are interested, are quantum extensions. Here, the eavesdropper is in possession of a system $E$ that extends the state $\rho_{\tilde{A}\tilde{B}\tilde{C}}$ shared by Alice, Bob, and Charlie in the sense that the extension state $\rho_{\tilde{A}\tilde{B}\tilde{C}E}$ satisfies  $\rho_{\tilde{A}\tilde{B}\tilde{C}}= \operatorname{Tr}_{E}[\rho_{\tilde{A}\tilde{B}\tilde{C}E}]$. A quantum extension of a correlation is defined as follows:
    \begin{align}
          & \rho_{ABCEXYZ} \notag \\
          & = \sum_{a,b,c,x,y,z}q(x,y,z)\ketbra*{abcxyz}{abcxyz}_{ABCXYZ}\otimes\text{Tr}_{ABC}\!\left[\left(\Pi^{(x)}_{a}\otimes\Pi^{(y)}_{b}\otimes\Pi^{(z)}_{c}\otimes\mathbb{I}_{E}\right)\rho_{\tilde{A}\tilde{B}\tilde{C}E}\right]\notag \\
            & =\sum_{a,b,c,x,y,z}q(x,y,z)\ketbra*{abcxyz}{abcxyz}_{ABCXYZ}\otimes p(a,b,c\vert x,y,z)\rho_{E}^{abcxyz}.\label{q_ext}
        \end{align}
        
    \begin{notation}    
    Henceforth, we employ the shorthand 
    \begin{equation}
     [abcxyz]_{ABCXYZ} \equiv    \ketbra*{abcxyz}{abcxyz}_{ABCXYZ},
    \end{equation}
     for the sake of brevity.
    \end{notation}
    
    The above no-signaling constraints and extensions, as well as quantum extensions, can be generalized to any multipartite scenario using the basic principle behind the no-signaling constraints. Appropriate no-signaling constraints apply when considering correlations involving multiple parties. Only after we have considered the no-signaling constraints can we begin to speak about what non-locality is and quantifying non-locality. To proceed, we need to define a quantity that can serve as a quantifier for multipartite non-locality.    
    
\section{Tripartite Intrinsic Non-Locality and its Properties}

\label{sec:tri-int-non}

\subsection{Conditional Total Correlation}

\label{subsec:tot-cor}

    In this subsection, we review the  conditional total correlation and its properties \cite{watanabe1960information} (see also \cite{AHS08,YHHHOS09}), before  defining our non-locality quantifier. We will discuss dual total correlation and its related non-locality quantifier in Section~\ref{sec:dual-total}.
    
    Total correlation is an $M$-partite generalization of mutual information. Conditional total correlation is the conditional version of total correlation, and it has previously been used in various multipartite scenarios in quantum information \cite{AHS08,YHHHOS09,7438836,li2018squashed}. Conditional total correlation of a multipartite state $\rho_{A_1 \cdots A_M E}$ is defined as
    \begin{align}
        I(A_{1};\cdots ;A_{M}\vert E) \coloneqq \sum_{i=1}^{M} H(A_{i}\vert E)- H(A_{1}\cdots A_{M}\vert E),
    \end{align}
    where $H(A\vert E) \coloneqq H(AE)-H(E)$, and $H(A) \coloneqq - \text{Tr}[\rho_{A}\log_{2}\rho_{A}]$. The chain rule for the bipartite conditional mutual information is as follows:
    \begin{align}\label{crule}
        I(A;BC\vert E)= I(A;B\vert CE)+I(A;C\vert E).
    \end{align}
    There exist chain rules for conditional total correlation \cite{watanabe1960information,YHHHOS09,li2018squashed}, which are as follows:
    \begin{align}
    \label{eqn}
        I(BA_{1};A_{2};\cdots ;A_{M}\vert E) & = I(A_{1};A_{2};\cdots ;A_{M}\vert BE) + \sum_{i=2}^{M} I(B;A_{i}\vert E),\\
    \label{eqn1}
        I(A_{1};\cdots;A_{M}\vert E) & = \sum_{j=1}^{M-1} I(A_{j};A_{j+1}\cdots A_{M}\vert E).
    \end{align}
    Let  $\rho_{A_1 \cdots A_M E}$ and $\sigma_{A_1 \cdots A_M E}$ be multipartite states, for which each of the subsystems $A_1$, \ldots, $A_M$ are finite-dimensional. Suppose that  $\frac{1}{2}\Vert\rho-\sigma\Vert_{1}\leq\varepsilon$, where $\varepsilon \in [0,1]$. Then the following uniform continuity bound holds \cite[Eq.~(60)]{shirokov2021uniform}:
    \begin{equation}\label{cont}
        \vert I(A_{1}; \cdots; A_{M} \vert E)_{\rho}-I(A_{1}; \cdots; A_{M} \vert E)_{\sigma}\vert \leq 2 \varepsilon \log_{2}\operatorname{dim} \mathcal{H}_{A_{1} \cdots A_{M-1}} + M g(\varepsilon),
    \end{equation}
    where
    \begin{align}
        \label{g}
        g(\varepsilon)\coloneqq (\varepsilon + 1) \log_2(\varepsilon + 1) - \varepsilon \log_2\varepsilon.
    \end{align}
    Conditional total correlation obeys data processing under local channels \cite{YHHHOS09}:
    \begin{align}
        I(A_{1};\cdots;A_{M}\vert E)_{\rho} \geq I(\tilde{A}_{1};\cdots;\tilde{A}_{M}\vert E)_{\omega} ,
    \end{align}
    where
    \begin{align} 
    \omega_{\tilde{A}_{1}\cdots\tilde{A}_{M}E} \coloneqq \left(\mathcal{N}^{(1)}_{A_{1}\rightarrow \tilde{A}_{1}}\otimes\cdots\otimes \mathcal{N}^{(M)}_{A_{M}\rightarrow\tilde{A}_{M}}\right) \left(\rho_{\tilde{A}_{1}\cdots\tilde{A}_{M}E}\right),
    \end{align}
    and $\mathcal{N}^{(i)}_{A_{i}\rightarrow \tilde{A}_{i}}$ is a channel, for $i \in \{1, \ldots, M\}$. We now define a first version of tripartite intrinsic non-locality.
    \begin{definition}
        Let  $p(a,b,c\vert x,y,z)$ be a no-signaling correlation. Tripartite intrinsic non-locality (TINL) of $p$ is defined as
        \begin{align}
            N(A;B;C)_{p} \coloneqq \frac{1}{2}\sup_{q(x,y,z)}\inf_{\rho_{ABCXYZE}}I(A;B;C\vert EXYZ)_{\rho},
        \end{align}
        where $q(x,y,z)$ is a probability distribution for the inputs of Alice, Bob, and Charlie and $\rho_{ABCXYZE}$ is a no-signaling extension of the state shared by Alice, Bob, and Charlie, given by
        \begin{align}\label{state}
            \rho_{ABCXYZE} = \sum_{a,b,c,x,y,z}q(x,y,z)p(a,b,c\vert x,y,z)[abcxyz]_{ABCXYZ}\otimes\rho_{E}^{abcxyz}.
        \end{align}        
    \end{definition}

    \begin{definition}
        Quantum tripartite intrinsic non-locality (QTINL) of a quantum correlation\\ $p(a,b,c\vert x,y,z)$ is defined as
        \begin{align}
            N_{Q}(A;B;C)_{p} \coloneqq \frac{1}{2}\sup_{q(x,y,z)}\inf_{\rho_{ABCXYZE}}I(A;B;C\vert EXYZ)_{\rho},
        \end{align}
        where $q(x,y,z)$ is a probability distribution for the inputs of Alice, Bob, and Charlie and $\rho_{ABCXYZE}$ is a quantum extension, as in \eqref{q_ext}, of the state shared by Alice, Bob, and Charlie, given by
        \begin{align}\label{qstate}
            \rho_{ABCXYZE} = \sum_{a,b,c,x,y,z}q(x,y,z)p(a,b,c\vert x,y,z)[abcxyz]_{ABCXYZ}\otimes\rho_{E}^{abcxyz}.
        \end{align}        
    \end{definition}

    The rest of this section is structured as follows. In Section~\ref{subsec:tri-chain}, we derive the chain rule that will help us prove further theorems about tripartite intrinsic non-locality and quantum tripartite intrinsic non-locality. In Section~\ref{subsec:tri-add}, we prove that tripartite intrinsic non-locality and quantum tripartite intrinsic non-locality are additive.     Additionally, we prove important properties of tripartite intrinsic non-locality and quantum tripartite intrinsic non-locality, such as convexity and monotonicity under local operations and common randomness in Appendices~\ref{subsec:tri-conv} and~\ref{subsec:tri-mono-locr}, respectively. We also prove in Appendix~\ref{subsec:local-cor} that tripartite intrinsic non-locality and quantum tripartite intrinsic non-locality vanish for local tripartite correlations. These results are important from a resource-theoretic perspective. 

\subsection{Chain Rule for Tripartite Conditional Total Correlation}\label{subsec:tri-chain}

    Before we can prove additivity and other important properties of tripartite intrinsic non-locality, we need to establish a chain rule for the conditional total correlation of two rounds of the conference key agreement protocol: \begin{align}I(A_{1}A_{2};B_{1}B_{2};C_{1}C_{2}\vert E).
    \end{align}
    We will resolve this quantity into a sum of conditional total correlation terms related to the individual rounds of the protocol and other information theoretic quantities that depend on both rounds. These extra information-theoretic quantities are expressed as conditional mutual information quantities. Later in Theorem~\ref{theorem_3}, we establish a general multipartite version of this chain rule.

    \begin{theorem}\label{theorem_2}
        For every  state $\rho_{A_{1}B_{1}C_{1}A_{2}B_{2}C_{2}E}$, the following equality holds:
        \begin{align}\label{eqnA}
            I(A_{1}&A_{2};B_{1}B_{2};C_{1}C_{2}\vert E)_{\rho}= I(A_{1};B_{1};C_{1}\vert EA_{2}B_{2}C_{2})_{\rho} + I(A_{2};B_{2};C_{2}\vert E)_{\rho} \nonumber\\
            &\qquad + I(C_{1};A_{2}B_{2}\vert EC_{2})_{\rho} + I(A_{1};B_{2}C_{2}\vert EA_{2})_{\rho} + I(B_{1};A_{2}C_{2}\vert EB_{2})_{\rho}   . 
        \end{align}
    \end{theorem}
    
    \begin{proof}
        Consider that, by applying definitions and the chain rule for conditional
entropy,
\begin{align}
& I(A_{1}A_{2};B_{1}B_{2};C_{1}C_{2}\vert E)\notag \\
& =H(A_{1}A_{2}\vert E)+H(B_{1}B_{2}\vert E)+H(C_{1}C_{2}\vert E)-H(A_{1}A_{2}B_{1}B_{2}
C_{1}C_{2}\vert E) \\
& =H(A_{2}\vert E)+H(A_{1}\vert EA_{2})+H(B_{2}\vert E)+H(B_{1}\vert EB_{2})+H(C_{2}
\vert E)+H(C_{1}\vert EC_{2}) \notag \\
& \qquad -H(A_{2}B_{2}C_{2}\vert E)-H(A_{1}B_{1}C_{1}\vert EA_{2}B_{2}C_{2})\\
& =I(A_{2};B_{2};C_{2}\vert E)+H(A_{1}\vert EA_{2})+H(B_{1}\vert EB_{2})+H(C_{1}
\vert EC_{2})-H(A_{1}B_{1}C_{1}\vert EA_{2}B_{2}C_{2}).
\end{align}
Then consider that
\begin{align}
& H(A_{1}\vert EA_{2})+H(B_{1}\vert EB_{2})+H(C_{1}\vert EC_{2})-H(A_{1}B_{1}C_{1}
\vert EA_{2}B_{2}C_{2})\notag \\
& =H(A_{1}\vert EA_{2})+H(B_{1}\vert EB_{2})+H(C_{1}\vert EC_{2})-H(A_{1}B_{1}C_{1}
\vert EA_{2}B_{2}C_{2})\notag \\
& \qquad +H(A_{1}\vert EA_{2}B_{2}C_{2})-H(A_{1}\vert EA_{2}B_{2}C_{2}) +H(B_{1}\vert EA_{2}B_{2}C_{2})-H(B_{1}\vert EA_{2}B_{2}C_{2})\notag \\
& \qquad +H(C_{1}\vert EA_{2}B_{2}C_{2})-H(C_{1}\vert EA_{2}B_{2}C_{2})\\
& =I(A_{1};B_{1};C_{1}\vert EA_{2}B_{2}C_{2})\notag  +H(A_{1}\vert EA_{2})-H(A_{1}\vert EA_{2}B_{2}C_{2})\notag \\
& \qquad +H(B_{1}\vert EB_{2})-H(B_{1}\vert EA_{2}B_{2}C_{2}) +H(C_{1}\vert EC_{2})-H(C_{1}\vert EA_{2}B_{2}C_{2})\\
& =I(A_{1};B_{1};C_{1}\vert EA_{2}B_{2}C_{2})+I(A_{1};B_{2}C_{2}\vert EA_{2}
)+I(B_{1};A_{2}C_{2}\vert EB_{2})+I(C_{1};A_{2}B_{2}\vert EC_{2}).
\end{align}
This concludes the proof. 
    \end{proof}

\subsection{Additivity}\label{subsec:tri-add}
    In this section, we prove that tripartite intrinsic non-locality is additive. This is indeed essential for the tripartite intrinsic non-locality to be a useful upper bound on DI conference key agreement capacity. 
    
    \begin{theorem}[Additivity of TINL]\label{theorem:additivity}
        Let $p(a_{1},a_{2},b_{1},b_{2},c_{1},c_{2}\vert x_{1},x_{2},y_{1},y_{2},z_{1},z_{2})$ be a no-signaling correlation for which  no-signaling constraints hold for all parties. For example, the no-signaling constraints for Alice are as follows:
        \begin{align}
            \sum_{a_{1}}p(a_{1},a_{2},b_{1},b_{2},c_{1},c_{2}&\vert x_{1},x_{2},y_{1},y_{2},z_{1},z_{2}) \nonumber\\ &= \sum_{a_{1}}p(a_{1},a_{2},b_{1},b_{2},c_{1},c_{2}\vert \bar{x}_{1},x_{2},y_{1},y_{2},z_{1},z_{2}) \quad \forall x_{1},\bar{x}_{1},\label{eq:intranosig1}\\
            \sum_{a_{2}}p(a_{1},a_{2},b_{1},b_{2},c_{1},c_{2}&\vert x_{1},x_{2},y_{1},y_{2},z_{1},z_{2}) \nonumber\\ &= \sum_{a_{2}}p(a_{1},a_{2},b_{1},b_{2},c_{1},c_{2}\vert x_{1},\bar{x}_{2},y_{1},y_{2},z_{1},z_{2}) \quad \forall x_{2},\bar{x}_{2}.\label{eq:intranosig2}
        \end{align}
        Suppose that similar constraints hold for Bob and Charlie as well.
        Let $t(a_{1},b_{1},c_{1}\vert x_{1},y_{1},z_{1})$ and $r(a_{2},b_{2},c_{2}\vert x_{2},y_{2},z_{2})$ be no-signaling correlations corresponding to the marginals of $p$. Then the intrinsic non-locality is superadditive, in the sense that
        \begin{align}
        \label{eq:superadd-ineq-1}
            N(A_{1}A_{2};B_{1}B_{2};C_{1}C_{2})_{p} \geq N(A_{1};B_{1};C_{1})_{t} + N(A_{2};B_{2};C_{2})_{r}.
        \end{align}
        If \begin{equation}
            p(a_{1},a_{2},b_{1},b_{2},c_{1},c_{2}\vert x_{1},x_{2},y_{1},y_{2},z_{1},z_{2}) = t(a_{1},b_{1},c_{1}\vert x_{1},y_{1},z_{1})r(a_{2},b_{2},c_{2}\vert x_{2},y_{2},z_{2}),
            \label{eq:prod-form}
        \end{equation}
        then the intrinsic non-locality is additive in the following sense:
        \begin{align}
            N(A_{1}A_{2};B_{1}B_{2};C_{1}C_{2})_{p} = N(A_{1};B_{1};C_{1})_{t} + N(A_{2};B_{2};C_{2})_{r}.
            \label{eq:additivity-int-non-loc}
        \end{align}
    \end{theorem}

    \noindent No-signaling constraints like \eqref{eq:intranosig1}--\eqref{eq:intranosig2} can in principle be enforced by a party performing parallel measurements shielded from each other, such as Alice recording $a_1$ and $a_2$ at separate locations between which communication is not possible. The stronger product assumption in \eqref{eq:prod-form} cannot be enforced in this way, but the condition will hold in the natural setting of sequential experimental trials in which an i.i.d.~assumption is made.
    
    \begin{proof}
        We first prove that tripartite intrinsic non-locality is superadditive in the sense of~\eqref{eq:superadd-ineq-1}, and then we prove it is subadditive when \eqref{eq:prod-form} holds. Additivity when \eqref{eq:prod-form} holds then follows as a consequence.
        
        First, let us prove superadditivity. To begin, let us consider states that arise from embedding an arbitrary no-signaling extension of $p(a_{1},a_{2},b_{1},b_{2},c_{1},c_{2}\vert x_{1},x_{2},y_{1},y_{2},z_{1},z_{2})$ into the following quantum state:
        \begin{multline}\label{ext}
            \zeta_{A_{1}B_{1}C_{1}A_{2}B_{2}C_{2}EX_{1}X_{2}Y_{1}Y_{2}Z_{1}Z_{2}} = \\ \sum_{\substack{a_{1},b_{1},c_{1},a_{2},b_{2},c_{2},\\
            x_{1},y_{1},z_{1},x_{2},y_{2},z_{2}}} q(x_{1},y_{1},z_{1},x_{2},y_{2},z_{2})p(a_{1},b_{1},c_{1},a_{2},b_{2},c_{2}\vert x_{1},y_{1},z_{1},x_{2},y_{2},z_{2})\\ [a_{1}b_{1}c_{1}a_{2}b_{2}c_{2}x_{1}y_{1}z_{1}x_{2}y_{2}z_{2}]_{A_{1}B_{1}C_{1}A_{2}B_{2}C_{2}X_{1}X_{2}Y_{1}Y_{2}Z_{1}Z_{2}}\otimes\rho_{E}^{a_{1}b_{1}c_{1}a_{2}b_{2}c_{2}x_{1}y_{1}z_{1}x_{2}y_{2}z_{2}}.
        \end{multline}
        We define the states $\tau$ and $\gamma$ to be the following arbitrary no-signaling extensions of $t$ and~$r$, respectively:
        \begin{multline}\label{ext_1}
            \tau_{A_{1}B_{1}C_{1}EX_{1}Y_{1}Z_{1}} = \sum_{a_{1},b_{1},c_{1},x_{1},y_{1},z_{1}} q(x_{1},y_{1},z_{1})t(a_{1},b_{1},c_{1}\vert x_{1},y_{1},z_{1})\\ [a_{1}b_{1}c_{1}x_{1}y_{1}z_{1}]_{A_{1}B_{1}C_{1}X_{1}Y_{1}Z_{1}}\otimes\rho_{E}^{a_{1}b_{1}c_{1}x_{1}y_{1}z_{1}},
        \end{multline}
        and
        \begin{multline}\label{ext_2}
            \gamma_{A_{2}B_{2}C_{2}EX_{2}Y_{2}Z_{2}} = \sum_{a_{2},b_{2},c_{2}, x_{2},y_{2},z_{2}} q(x_{2},y_{2},z_{2})r(a_{2},b_{2},c_{2}\vert x_{2},y_{2},z_{2})\\ [a_{2}b_{2}c_{2}x_{2}y_{2}z_{2}]_{A_{2}B_{2}C_{2}X_{2}Y_{2}Z_{2}}\otimes\rho_{E}^{a_{2}b_{2}c_{2}x_{2}y_{2}z_{2}}.
        \end{multline}
        Now, we use the chain rule from Theorem~\ref{theorem_2} to conclude that
        \begin{align}
            &I(A_{1}A_{2};B_{1}B_{2};C_{1}C_{2}\vert EX_{1}X_{2}Y_{1}Y_{2}Z_{1}Z_{2})_{\zeta}\nonumber\\
            &= I(A_{1};B_{1};C_{1}\vert EX_{1}X_{2}Y_{1}Y_{2}Z_{1}Z_{2}A_{2}B_{2}C_{2})_{\zeta} + I(A_{2};B_{2};C_{2}\vert EX_{1}X_{2}Y_{1}Y_{2}Z_{1}Z_{2})_{\zeta} \nonumber\\
            &\qquad+ I(A_{2}B_{2};C_{1}\vert EX_{1}X_{2}Y_{1}Y_{2}Z_{1}Z_{2}C_{2})_{\zeta} + I(B_{2}C_{2};A_{1}\vert EX_{1}X_{2}Y_{1}Y_{2}Z_{1}Z_{2}A_{2})_{\zeta} \nonumber\\
            &\qquad \qquad+ I(A_{2}C_{2};B_{1}\vert EX_{1}X_{2}Y_{1}Y_{2}Z_{1}Z_{2}B_{2})_{\zeta}
        \end{align}
        Since conditional mutual information is always non-negative, we conclude that
        \begin{multline}
            I(A_{1}A_{2};B_{1}B_{2};C_{1}C_{2}\vert EX_{1}X_{2}Y_{1}Y_{2}Z_{1}Z_{2})_{\zeta}\\
            \geq I(A_{1};B_{1};C_{1}\vert EX_{1}X_{2}Y_{1}Y_{2}Z_{1}Z_{2}A_{2}B_{2}C_{2})_{\zeta} + I(A_{2};B_{2};C_{2}\vert EX_{1}X_{2}Y_{1}Y_{2}Z_{1}Z_{2})_{\zeta}.  
        \end{multline}
        The state $\zeta_{A_{1}B_{1}C_{1}A_{2}B_{2}C_{2}EX_{1}X_{2}Y_{1}Y_{2}Z_{1}Z_{2}}$ is a valid no-signaling extension of  $t$ with extension systems $E X_2 Y_2 Z_2 A_2 B_2 C_2$, and the state $\zeta_{A_{2}B_{2}C_{2}EX_{1}X_{2}Y_{1}Y_{2}Z_{1}Z_{2}}$ is a valid no-signaling extension of $r$ with extension systems $E X_1 Y_1 Z_1$. So we conclude that
        \begin{align}
            &I(A_{1}A_{2};B_{1}B_{2};C_{1}C_{2}\vert EX_{1}X_{2}Y_{1}Y_{2}Z_{1}Z_{2})_{\zeta}\nonumber\\
            &\geq I(A_{1};B_{1};C_{1}\vert EX_{1}X_{2}Y_{1}Y_{2}Z_{1}Z_{2}A_{2}B_{2}C_{2})_{\zeta} + I(A_{2};B_{2};C_{2}\vert EX_{1}X_{2}Y_{1}Y_{2}Z_{1}Z_{2})_{\zeta}\\
            &\geq \inf_{\text{ext.~in }\eqref{ext_1}} I(A_{1};B_{1};C_{1}\vert EX_{1}Y_{1}Z_{1})_{\tau} + \inf_{\text{ext.~in }\eqref{ext_2}} I(A_{2};B_{2};C_{2}\vert EX_{2}Y_{2}Z_{2})_{\gamma}.   
        \end{align}
        Since the state $\zeta_{A_{1}B_{1}C_{1}A_{2}B_{2}C_{2}EX_{1}X_{2}Y_{1}Y_{2}Z_{1}Z_{2}}$ is an arbitrary no-signaling extension of $p$, we conclude that
        \begin{multline}
            \inf_{\text{ext.~in }\eqref{ext}} I(A_{1}A_{2};B_{1}B_{2};C_{1}C_{2}\vert EX_{1}X_{2}Y_{1}Y_{2}Z_{1}Z_{2})_{\zeta}
            \\
            \geq \inf_{\text{ext.~in }\eqref{ext_1}} I(A_{1};B_{1};C_{1}\vert EX_{1}Y_{1}Z_{1})_{\tau} + \inf_{\text{ext.~in }\eqref{ext_2}} I(A_{2};B_{2};C_{2}\vert EX_{2}Y_{2}Z_{2})_{\gamma}.  
        \end{multline}
        By optimizing over product input probability distributions, we have that
        \begin{multline}
            \sup_{q(x_1,y_1,z_1)q(x_2,y_2,z_2)} \inf_{\text{ext.~in }\eqref{ext}} I(A_{1}A_{2};B_{1}B_{2};C_{1}C_{2}\vert EX_{1}X_{2}Y_{1}Y_{2}Z_{1}Z_{2})_{\zeta}\\
            \geq \sup_{q(x_1,y_1,z_1)} \inf_{\text{ext.~in }\eqref{ext_1}} I(A_{1};B_{1};C_{1}\vert EX_{1}Y_{1}Z_{1})_{\tau} +\\ \sup_{q(x_2,y_2,z_2)} \inf_{\text{ext.~in }\eqref{ext_2}} I(A_{2};B_{2};C_{2}\vert EX_{2}Y_{2}Z_{2})_{\gamma}.  
        \end{multline}
        Hence, by optimizing the left-hand side over all input probability distributions, we conclude that
        \begin{align}\label{sup}
            N(A_{1}A_{2};B_{1}B_{2};C_{1}C_{2})_{p} \geq N(A_{1};B_{1};C_{1})_{t} + N(A_{2};B_{2};C_{2})_{r}.
        \end{align}
        This concludes the proof of superadditivity (i.e., the proof of \eqref{eq:superadd-ineq-1}).

        Let us prove subadditivity when \eqref{eq:prod-form} holds; i.e., let us prove that
        \begin{align}
            N(A_{1}A_{2};B_{1}B_{2};C_{1}C_{2})_{p} \leq N(A_{1};B_{1};C_{1})_{t} + N(A_{2};B_{2};C_{2})_{r}.
        \end{align}
        Consider the following quantum embeddings:
        \begin{multline}\label{ext1}
            \zeta_{A_{1}B_{1}C_{1}A_{2}B_{2}C_{2}EX_{1}X_{2}Y_{1}Y_{2}Z_{1}Z_{2}} =\\ \sum_{\substack{a_{1},b_{1},c_{1},a_{2},b_{2},c_{2},\\x_{1},y_{1},z_{1},x_{2},y_{2},z_{2}}} q(x_{1},y_{1},z_{1},x_{2},y_{2},z_{2})t(a_{1},b_{1},c_{1}\vert x_{1},y_{1},z_{1}) r(a_{2},b_{2},c_{2}\vert x_{2},y_{2},z_{2})\\ [a_{1}b_{1}c_{1}a_{2}b_{2}c_{2}x_{1}y_{1}z_{1}x_{2}y_{2}z_{2}]_{A_{1}B_{1}C_{1}A_{2}B_{2}C_{2}X_{1}X_{2}Y_{1}Y_{2}Z_{1}Z_{2}}\otimes\zeta_{E}^{a_{1}b_{1}c_{1}x_{1}y_{1}z_{1}a_{2}b_{2}c_{2}x_{2}y_{2}z_{2}}
        \end{multline}
        \begin{multline}\label{ext6}
            \rho_{A_{1}B_{1}C_{1}A_{2}B_{2}C_{2}X_{1}X_{2}Y_{1}Y_{2}Z_{1}Z_{2}E_{1}E_{2}} =\\ \sum_{\substack{a_{1},b_{1},c_{1},a_{2},b_{2},c_{2},\\x_{1},y_{1},z_{1},x_{2},y_{2},z_{2}}} q(x_{1},y_{1},z_{1},x_{2},y_{2},z_{2})t(a_{1},b_{1},c_{1}\vert x_{1},y_{1},z_{1}) r(a_{2},b_{2},c_{2}\vert x_{2},y_{2},z_{2})\\ [a_{1}b_{1}c_{1}a_{2}b_{2}c_{2}x_{1}y_{1}z_{1}x_{2}y_{2}z_{2}]_{A_{1}B_{1}C_{1}A_{2}B_{2}C_{2}X_{1}X_{2}Y_{1}Y_{2}Z_{1}Z_{2}}\otimes\rho_{E_{1}}^{a_{1}b_{1}c_{1}x_{1}y_{1}z_{1}}\otimes\rho_{E_{2}}^{a_{2}b_{2}c_{2}x_{2}y_{2}z_{2}},
        \end{multline}
        \begin{multline}\label{ext2}
            \tau_{A_{1}B_{1}C_{1}EX_{1}Y_{1}Z_{1}} =\\ \sum_{a_{1},b_{1},c_{1},x_{1},y_{1},z_{1}} q(x_{1},y_{1},z_{1})t(a_{1},b_{1},c_{1}\vert x_{1},y_{1},z_{1})[a_{1}b_{1}c_{1}x_{1}y_{1}z_{1}]_{A_{1}B_{1}C_{1}X_{1}Y_{1}Z_{1}}\otimes\rho_{E_{1}}^{a_{1}b_{1}c_{1}x_{1}y_{1}z_{1}},
        \end{multline}
        and 
        \begin{multline}\label{ext3}
            \gamma_{A_{2}B_{2}C_{2}EX_{2}Y_{2}Z_{2}} =\\ \sum_{a_{2},b_{2},c_{2},x_{2},y_{2},z_{2}} q(x_{2},y_{2},z_{2})r(a_{2},b_{2},c_{2}\vert x_{2},y_{2},z_{2})[a_{2}b_{2}c_{2}x_{2}y_{2}z_{2}]_{A_{2}B_{2}C_{2}X_{2}Y_{2}Z_{2}}\otimes\rho_{E_{2}}^{a_{2}b_{2}c_{2}x_{2}y_{2}z_{2}}.
        \end{multline}
        All the extensions above are no-signaling extensions. Consider that
        \begin{multline}
            \inf_{\text{ext.~in }\eqref{ext1}} I(A_{1}A_{2};B_{1}B_{2};C_{1}C_{2}\vert EX_{1}X_{2}Y_{1}Y_{2}Z_{1}Z_{2})_{\zeta}\\
            \leq I(A_{1}A_{2};B_{1}B_{2};C_{1}C_{2}\vert E_{1}E_{2}X_{1}X_{2}Y_{1}Y_{2}Z_{1}Z_{2})_{\rho}.
        \end{multline}
        Using the chain rule from Theorem~\ref{theorem_2}, we find that
        \begin{align}
        & I(A_{1}A_{2};B_{1}B_{2};C_{1}C_{2}\vert E_{1}E_{2}X_{1}X_{2}Y_{1}Y_{2}Z_{1}Z_{2})_{\rho} \notag \\
            &= I(A_{1};B_{1};C_{1}\vert E_{1}E_{2}X_{1}X_{2}Y_{1}Y_{2}Z_{1}Z_{2}A_{2}B_{2}C_{2})_{\rho} + I(A_{2};B_{2};C_{2}\vert E_{1}E_{2}X_{1}X_{2}Y_{1}Y_{2}Z_{1}Z_{2})_{\rho} \nonumber\\
            &\qquad + I(A_{2}B_{2};C_{1}\vert E_{1}E_{2}X_{1}X_{2}Y_{1}Y_{2}Z_{1}Z_{2}C_{2})_{\rho} + I(B_{2}C_{2};A_{1}\vert E_{1}E_{2}X_{1}X_{2}Y_{1}Y_{2}Z_{1}Z_{2}A_{2})_{\rho} \nonumber\\
            &\qquad \qquad + I(A_{2}C_{2};B_{1}\vert E_{1}E_{2}X_{1}X_{2}Y_{1}Y_{2}Z_{1}Z_{2}B_{2})_{\rho}.
        \end{align}
        We can write $I(A_{2}C_{2};B_{1}\vert E_{1}E_{2}X_{1}X_{2}Y_{1}Y_{2}Z_{1}Z_{2}B_{2})_{\rho}$ as follows:
        \begin{align}
            &I(A_{2}C_{2};B_{1}\vert E_{1}E_{2}X_{1}X_{2}Y_{1}Y_{2}Z_{1}Z_{2}B_{2})_{\rho}\notag \\
            &= H(A_{2}C_{2}\vert E_{1}E_{2}X_{1}X_{2}Y_{1}Y_{2}Z_{1}Z_{2}B_{2})_{\rho} -  H(A_{2}C_{2}\vert E_{1}E_{2}X_{1}X_{2}Y_{1}Y_{2}Z_{1}Z_{2}B_{2}B_{1})_{\rho}\\
            &= \sum_{x_{1},y_{1},z_{1},x_{2},y_{2},z_{2}} p(x_{1},y_{1},z_{1},x_{2},y_{2},z_{2})[H(A_{2}C_{2}\vert E_{1}E_{2}B_{2})_{\eta^{x_{1}y_{1}z_{1}x_{2}y_{2}z_{2}}} \notag \\  &\qquad\qquad\qquad\qquad\qquad\qquad\qquad\qquad\qquad -H(A_{2}C_{2}\vert E_{1}E_{2}B_{2}B_{1})_{\eta^{x_{1}y_{1}z_{1}x_{2}y_{2}z_{2}}}],
        \end{align}
        where, due to the no-signaling constraints on $p$, $t$, and $r$, we can write 
        \begin{equation}
        \eta^{x_{1}y_{1}z_{1}x_{2}y_{2}z_{2}}_{B_{1}A_{2}B_{2}C_{2}E_{1}E_{2}} = \sum_{b_{1}} t(b_{1}\vert y_{1})[b_{1}]\otimes\rho_{E_{1}}^{b_{1} y_{1}}\otimes\sum_{a_{2},b_{2},c_{2}} r(a_{2},b_{2},c_{2}\vert x_{2},y_{2},z_{2})[a_{2}b_{2}c_{2}]\otimes\rho_{E_{2}}^{a_{2}b_{2}c_{2}x_{2}y_{2}z_{2}},
            \end{equation}
            and
            \begin{equation}            \eta^{x_{1}y_{1}z_{1}x_{2}y_{2}z_{2}}_{A_{2}B_{2}C_{2}E_{1}E_{2}} = \sum_{a_{2},b_{2},c_{2}} r(a_{2},b_{2},c_{2}\vert x_{2},y_{2},z_{2})[a_{2}b_{2}c_{2}]\otimes\rho_{E_{2}}^{a_{2}b_{2}c_{2}x_{2}y_{2}z_{2}}\otimes\rho_{E_{1}},   
            \end{equation}
        where
        \begin{align}
        \rho^{b_1  y_1 }_{E_{1}} & = \sum_{a_1,  c_1} t(a_1,b_1, c_1\vert x_1, y_1, z_1)\rho^{a_1 b_1 c_1 x_1 y_1 z_1}_{E_{1}},\\
        \rho_{E_{1}} & = \sum_{a_1, b_1, c_1} t(a_1,b_1, c_1\vert x_1, y_1, z_1)\rho^{a_1 b_1 c_1 x_1 y_1 z_1}_{E_{1}}.    
        \end{align}
        From the above definitions, we can conclude that
        \begin{align}
            H(A_{2}C_{2}\vert E_{1}E_{2}B_{2}B_{1})_{\eta^{x_{1}y_{1}z_{1}x_{2}y_{2}z_{2}}} = H(A_{2}C_{2}\vert E_{1}E_{2}B_{2})_{\eta^{x_{1}y_{1}z_{1}x_{2}y_{2}z_{2}}}.
        \end{align}
        Hence,
        \begin{multline}
            I(A_{2}C_{2};B_{1}\vert E_{1}E_{2}X_{1}X_{2}Y_{1}Y_{2}Z_{1}Z_{2}B_{2})_{\rho} \\
            = H(A_{2}C_{2}\vert E_{1}E_{2}X_{1}X_{2}Y_{1}Y_{2}Z_{1}Z_{2}B_{2})_{\rho} -  H(A_{2}C_{2}\vert E_{1}E_{2}X_{1}X_{2}Y_{1}Y_{2}Z_{1}Z_{2}B_{2})_{\rho}= 0.
        \end{multline}
        The quantities $I(B_{2}C_{2};A_{1}\vert E_{1}E_{2}X_{1}X_{2}Y_{1}Y_{2}Z_{1}Z_{2}A_{2})_{\rho}$ and $I(A_{2}C_{2};B_{1}\vert E_{1}E_{2}X_{1}X_{2}Y_{1}Y_{2}Z_{1}Z_{2}B_{2})_{\rho}$ are equal to zero using similar arguments. This leads to the following conclusion:
        \begin{align}
            &\inf_{\text{ext.~in }\eqref{ext1}}I(A_{1}A_{2};B_{1}B_{2};C_{1}C_{2}\vert EX_{1}X_{2}Y_{1}Y_{2}Z_{1}Z_{2})_{\zeta}\nonumber\\
            &\leq I(A_{1}A_{2};B_{1}B_{2};C_{1}C_{2}\vert E_{1}E_{2}X_{1}X_{2}Y_{1}Y_{2}Z_{1}Z_{2})_{\rho}\nonumber\\
            &=I(A_{1};B_{1};C_{1}\vert E_{1}E_{2}X_{1}X_{2}Y_{1}Y_{2}Z_{1}Z_{2}A_{2}B_{2}C_{2})_{\rho} + I(A_{2};B_{2};C_{2}\vert E_{1}E_{2}X_{1}X_{2}Y_{1}Y_{2}Z_{1}Z_{2})_{\rho} \notag \\
            & = 
            I(A_{1};B_{1};C_{1}\vert E_{1}X_{1}Y_{1}Z_{1})_{\tau} + I(A_{2};B_{2};C_{2}\vert E_{2}X_{2}Y_{2}Z_{2})_{\gamma},
        \end{align}
        where the last line follows from  the structure of the state in \eqref{ext6} and the fact that the extension is a no-signaling extension.
        Since the no-signaling extensions $\tau$ and $\gamma$ are arbitrary, we conclude that
        \begin{multline}
            \inf_{\text{ext.~in }\eqref{ext1}} I(A_{1}A_{2};B_{1}B_{2};C_{1}C_{2}\vert EX_{1}X_{2}Y_{1}Y_{2}Z_{1}Z_{2})_{\zeta}\\
            \leq \inf_{\text{ext.~in }\eqref{ext2}} I(A_{1};B_{1};C_{1}\vert EX_{1}Y_{1}Z_{1})_{\tau} + \inf_{\text{ext.~in }\eqref{ext3}} I(A_{2};B_{2};C_{2}\vert EX_{2}Y_{2}Z_{2})_{\gamma}.
        \end{multline}
        Now optimizing over arbitrary input probability distributions, we find that
        \begin{multline}
            \sup_{q} \inf_{\text{ext.~in }\eqref{ext1}} I(A_{1}A_{2};B_{1}B_{2};C_{1}C_{2}\vert EX_{1}X_{2}Y_{1}Y_{2}Z_{1}Z_{2})_{\zeta}\\
            \leq\sup_{q} \inf_{\text{ext.~in }\eqref{ext2}} I(A_{1};B_{1};C_{1}\vert EX_{1}Y_{1}Z_{1})_{\tau} + \sup_{q} \inf_{\text{ext.~in }\eqref{ext3}} I(A_{2};B_{2};C_{2}\vert EX_{2}Y_{2}Z_{2})_{\gamma}.
        \end{multline}
        Hence,
        \begin{equation}\label{sub}
            N(A_{1}A_{2};B_{1}B_{2};C_{1}C_{2})_{p} \leq N(A_{1};B_{1};C_{1})_{t} + N(A_{2};B_{2};C_{2})_{r}.
        \end{equation}
        Putting together \eqref{sup} and \eqref{sub}, we have established additivity (i.e., we have proven \eqref{eq:additivity-int-non-loc}).
    \end{proof}

    \begin{theorem}[Additivity of QTINL]\label{theorem:additivity-q}
        Let $p(a_{1},a_{2},b_{1},b_{2},c_{1},c_{2}\vert x_{1},x_{2},y_{1},y_{2},z_{1},z_{2})$ be a quantum correlation for which  no-signaling constraints hold for all parties. For example, the no-signaling constraints for Alice are as follows:
        \begin{align}
            \sum_{a_{1}}p(a_{1},a_{2},b_{1},b_{2},c_{1},c_{2}&\vert x_{1},x_{2},y_{1},y_{2},z_{1},z_{2}) \nonumber\\ &= \sum_{a_{1}}p(a_{1},a_{2},b_{1},b_{2},c_{1},c_{2}\vert \bar{x}_{1},x_{2},y_{1},y_{2},z_{1},z_{2}) \quad \forall x_{1},\bar{x}_{1},\\
            \sum_{a_{2}}p(a_{1},a_{2},b_{1},b_{2},c_{1},c_{2}&\vert x_{1},x_{2},y_{1},y_{2},z_{1},z_{2}) \nonumber\\ &= \sum_{a_{2}}p(a_{1},a_{2},b_{1},b_{2},c_{1},c_{2}\vert x_{1},\bar{x}_{2},y_{1},y_{2},z_{1},z_{2}) \quad \forall x_{2},\bar{x}_{2}.
        \end{align}
        Suppose that similar constraints hold for Bob and Charlie as well.
        Let $t(a_{1},b_{1},c_{1}\vert x_{1},y_{1},z_{1})$ and $r(a_{2},b_{2},c_{2}\vert x_{2},y_{2},z_{2})$ be quantum correlations corresponding to the marginals of $p$. Then the quantum intrinsic non-locality is superadditive, in the sense that
        \begin{align}
        \label{eq:q-superadd-ineq-1}
            N_Q(A_{1}A_{2};B_{1}B_{2};C_{1}C_{2})_{p} \geq N_Q(A_{1};B_{1};C_{1})_{t} + N_Q(A_{2};B_{2};C_{2})_{r}.
        \end{align}
        If \begin{equation}
            p(a_{1},a_{2},b_{1},b_{2},c_{1},c_{2}\vert x_{1},x_{2},y_{1},y_{2},z_{1},z_{2}) = t(a_{1},b_{1},c_{1}\vert x_{1},y_{1},z_{1})r(a_{2},b_{2},c_{2}\vert x_{2},y_{2},z_{2}),
            \label{eq:q-prod-form}
        \end{equation}
        then the quantum intrinsic non-locality is additive in the following sense:
        \begin{align}
            N_Q(A_{1}A_{2};B_{1}B_{2};C_{1}C_{2})_{p} = N_Q(A_{1};B_{1};C_{1})_{t} + N_Q(A_{2};B_{2};C_{2})_{r}.
            \label{eq:additivity-q-int-non-loc}
        \end{align}
    \end{theorem}
    \begin{proof}
        The proof follows by using similar techniques as Theorem~\ref{theorem:additivity} and by taking appropriate quantum extensions.
    \end{proof}

    \section{Multipartite Intrinsic Non-Locality}

    \label{sec:mul-int-non}
    
        We now generalize the tripartite case to the multipartite case. Henceforth, we denote the $i$th input to the measurement device by $x_{i}$, and we denote the outcome of a measurement by $a_{i}$, where $i\in\{1, \ldots, M\}$ and $M$ is the number of parties involved.         Now, we can define multipartite intrinsic non-locality, using conditional total correlation, for a no-signaling correlation as follows:
        \begin{definition}
            Let $p(a_{1},\ldots,a_{M}\vert x_{1},\ldots,x_{M})$ be a no-signaling correlation. Multipartite intrinsic non-locality of $p$ is defined as
            \begin{align}
                N(A_{1};\cdots;A_{M})_{p} \coloneqq \frac{1}{M-1}\sup_{q(x_{1},\ldots,x_{M})}\inf_{\rho_{A_{1}\cdots A_{M}X_{1}\cdots X_{M}E}}I(A_{1};\cdots;A_{M}\vert E X_{1}\cdots X_{M})_{\rho},
            \end{align}
            where $q(x_{1},\ldots,x_{M})$ is a probability distribution for the inputs of the Alices, and the state $\rho_{A_{1}\cdots A_{M}X_{1}\cdots X_{M}E}$ is a no-signaling extension of the state shared by the Alices, given by
            \begin{multline}\label{statemult}
                \rho_{A_{1}\cdots A_{M}X_{1}\cdots X_{M}} = \sum_{a_{1},\ldots,a_{M},x_{1},\ldots,x_{M}}q(x_{1},\ldots,x_{M})p(a_{1},\ldots,a_{M}\vert x_{1},\ldots,x_{M})\\
                [a_{1},\ldots,a_{M},x_{1},\ldots,x_{M}]_{A_{1}\cdots A_{M}X_{1} \cdots X_{M}}\otimes\rho_{E}^{a_{1},\ldots,a_{M},x_{1},\ldots,x_{M}}.
            \end{multline}        
        \end{definition}
        
        We define quantum multipartite quantum intrinsic non-locality, based on conditional total correlation, for a quantum correlation as follows:
        
        \begin{definition}
            Multipartite quantum intrinsic non-locality of  $p(a_{1},\ldots,a_{M}\vert x_{1},\ldots,x_{M})$, a quantum correlation, is defined as
            \begin{align}
                N_{Q}(A_{1};\cdots;A_{M})_{p} \coloneqq  \frac{1}{M-1}\sup_{q(x_{1},\ldots,x_{M})}\inf_{\rho_{A_{1}\cdots A_{M}X_{1}\cdots X_{M}E}}I(A_{1};\cdots;A_{M}\vert E X_{1}\cdots X_{M})_{\rho},
            \end{align}
            where $q(x_{1},\ldots,x_{M})$ is a probability distribution for generating the inputs used by the Alices and $\rho_{A_{1}\cdots A_{M}X_{1}\cdots X_{M}E}$ is a quantum extension of the state shared by Alices, given by
            \begin{multline}\label{statemult-q}
                \rho_{A_{1}\cdots A_{M}X_{1}\cdots X_{M}} = \sum_{a_{1},\ldots,a_{M},x_{1},\ldots,x_{M}}q(x_{1},\ldots,x_{M})p(a_{1},\ldots,a_{M}\vert x_{1},\ldots,x_{M})\\
                [a_{1},\ldots,a_{M},x_{1},\ldots,x_{M}]_{A_{1}\cdots A_{M}X_{1}\cdots X_{M}}\otimes\rho_{E}^{a_{1},\ldots,a_{M},x_{1},\ldots,x_{M}}.
            \end{multline}        
        \end{definition}
        
        We now derive a chain rule for the quantity $I(A_{1,1}A_{1,2};\cdots;A_{i,1}A_{i,2};\cdots;A_{M,1}A_{M,2}\vert E)$ similar to that in Theorem~\ref{theorem_2}. In doing so, we generalize \eqref{eqnA} to every finite $M$ such that we can  prove additivity and other relevant properties of multipartite (quantum) intrinsic non-locality. Let us define $[M] \coloneqq  \{1,2, \ldots ,m\}$ and  $A_{\{i, \ldots, M\},j} \equiv A_{i,j} \cdots A_{M,j}$. 
        \begin{theorem}\label{theorem_3}
            For every multipartite state $\rho_{A_{1,1}A_{1,2}\cdots A_{i,1}A_{i,2}\cdots A_{M,1}A_{M,2}E}$, the following equality holds:
            \begin{multline}\label{crulemult}
                I(A_{1,1}A_{1,2};\cdots;A_{i,1}A_{i,2};\cdots;A_{M,1}A_{M,2}\vert E) 
                = \\
                I(A_{1,2};\cdots;A_{M,2}\vert E) + I(A_{1,1};\cdots;A_{M,1}\vert EA_{[M],2})+\sum_{i=1}^{M} I(A_{i,1};A_{[M]\backslash \{i\},2}\vert EA_{i,2}).
            \end{multline}
        \end{theorem}
        
        \begin{proof}
        By applying definitions and the chain rule for conditional entropy, we find that
        \begin{align}
            & I(A_{1,1}A_{1,2};A_{2,1}A_{2,2};\cdots;A_{M,1}A_{M,2}\vert E)\nonumber\\
            & =\sum_{i=1}^{M}H(A_{i,1}A_{i,2}\vert E)-H(A_{1,1}A_{1,2}A_{2,1}A_{2,2}\cdots
                A_{M,1}A_{M,2}\vert E)\\
            & =\sum_{i=1}^{M}\left[  H(A_{i,2}\vert E)+H(A_{i,1}\vert EA_{i,2})\right]  \nonumber\\
            & \qquad-\left[  H(A_{1,2}A_{2,2}\cdots A_{M,2}\vert E)-H(A_{1,1}A_{2,1}\cdots
            A_{M,1}\vert EA_{1,2}A_{2,2}\cdots A_{M,2})\right]  \\
            & =I(A_{1,2};A_{2,2};\cdots;A_{M,2}\vert E) +\sum_{i=1}^{M}H(A_{i,1}\vert EA_{i,2}
            )-H(A_{1,1}A_{2,1}\cdots A_{M,1}\vert EA_{[M],2}).
        \end{align}
        Continuing, we find that
        \begin{align}
            & \sum_{i=1}^{M}H(A_{i,1}\vert EA_{i,2})-H(A_{1,1}A_{2,1}\cdots A_{M,1}
            \vert E A_{[M],2})\nonumber\\
            & =\sum_{i=1}^{M} \left[H(A_{i,1}\vert EA_{i,2})-H(A_{i,1}\vert EA_{\left[  M\right],2})+H(A_{i,1}\vert EA_{\left[  M\right]  ,2})\right]\nonumber\\
            & \qquad-H(A_{1,1}A_{2,1}\cdots A_{M,1}\vert EA_{[M],2})\\
            & =\sum_{i=1}^{M}I(A_{i,1};A_{\left[  M\right]  \backslash\left\{  i\right\},2}\vert EA_{i,2})+I(A_{1,1};A_{2,1};\cdots;A_{M,1}\vert EA_{[M],2}).
        \end{align}
    This concludes the proof.
            \end{proof}

        Now, let us note that if we consider the particular case when $M=3$, we recover the exact form obtained earlier in \eqref{eqnA}. Then, we can extend the arguments presented for the tripartite case to obtain additivity, convexity, and monotonicity under LOCR for multipartite intrinsic non-locality and multipartite quantum intrinsic non-locality, primarily due to the structure of \eqref{crulemult} producing similar terms for every finite $M$. 
        
    \section{Dual Multipartite Intrinsic Non-Locality}
    
    \label{sec:dual-total}
    
        Until now, we have defined multipartite intrinsic non-locality based on conditional total correlation. As noted earlier, total correlation is just one possible generalization of mutual information that has found uses in quantum information. Dual total correlation is another $M$-partite generalization of mutual information, first introduced in \cite{Han75,Han78}. Both total correlation and dual total correlation correspond to mutual information for the bipartite scenario. Since a distinction between total correlation and dual total correlation would only arise in the multipartite scenario, it is worthwhile to discuss the multipartite intrinsic non-locality based on conditional dual total correlation to note the differences in quantities that arise and compare the two quantities.
        
        In this section, we discuss multipartite intrinsic non-locality based on dual total correlation. Conditional dual total correlation is the conditional version of dual total correlation, and it has been previously used in various multipartite scenarios in quantum information \cite{YHHHOS09,PhysRevLett.101.140501}. Conditional dual total correlation of a state $\rho_{A_{1}\cdots  A_{M}E}$ is defined as
        
        \begin{equation}
            \widetilde{I}(A_{1};\cdots;A_{M}\vert E) \coloneqq \sum_{i=1}^{m}H(A_{[M]\backslash \{i\}}\vert E)-(m-1)H(A_{1}\cdots A_{M}\vert E).
        \end{equation}
        The chain rule for conditional dual total correlation is as follows:
        \begin{equation}
            \widetilde{I}(BA_{1};A_{2};\cdots;A_{M}\vert E) = \widetilde{I}(A_{1};A_{2};\cdots;A_{M}\vert BE) + I(B;A_{2}\cdots A_{M}\vert E).
        \end{equation}
        We now define the multipartite intrinsic non-locality based on conditional dual total correlation, and we refer to it as dual multipartite intrinsic non-locality:
        \begin{definition}
            Dual multipartite intrinsic non-locality of a no-signaling correlation  \\ $p(a_{1},\ldots,a_{M}\vert x_{1},\ldots,x_{M})$ is defined as
            \begin{align}
                \widetilde{N}(A_{1};\cdots;A_{M})_{p} \coloneqq \sup_{q(x_{1},\ldots,x_{M})}\inf_{\rho_{A_{1}\cdots A_{M}X_{1} \cdots X_{M}E}}\widetilde{I}(A_{1};\cdots;A_{M}\vert E X_{1}\cdots X_{M})_{\rho},
            \end{align}
            where $q(x_{1},\ldots,x_{M})$ is a probability distribution for the inputs of the Alices, and the state $\rho_{A_{1} \cdots A_{M}X_{1} \cdots X_{M}E}$ is a no-signaling extension of the state shared by the Alices, given by
            \begin{multline}
                \rho_{A_{1}\cdots A_{M}X_{1}\cdots X_{M}} = \sum_{a_{1},\ldots,a_{M},x_{1},\ldots,x_{M}}q(x_{1},\ldots,x_{M})p(a_{1},\ldots,a_{M}\vert x_{1},\ldots,x_{M})\\
                [a_{1},\ldots,a_{M},x_{1},\ldots,x_{M}]_{A_{1} \cdots  A_{M}X_{1} \cdots X_{M}}\otimes\rho_{E}^{a_{1},\ldots,a_{M},x_{1},\ldots,x_{M}}.
            \end{multline}        
        \end{definition}
        
        We define dual multipartite quantum intrinsic non-locality for a quantum correlation as follows:
        
        \begin{definition}
            Dual multipartite quantum intrinsic non-locality of  $p(a_{1},\ldots,a_{M}\vert x_{1},\ldots,x_{M})$, a quantum correlation, is defined as
            \begin{align}
                \widetilde{N}_{Q}(A_{1};\cdots;A_{M})_{p} \coloneqq  \sup_{q(x_{1},\ldots,x_{M})}\inf_{\rho_{A_{1}\cdots A_{M}X_{1}\cdots X_{M}E}}\widetilde{I}(A_{1};\cdots;A_{M}\vert E X_{1}\cdots X_{M})_{\rho},
            \end{align}
            where $q(x_{1},\ldots,x_{M})$ is a probability distribution that generates the inputs used by the Alices and $\rho_{A_{1}\cdots A_{M}X_{1} \cdots X_{M}E}$ is a quantum extension of the state shared by Alices, given by
            \begin{multline}
                \rho_{A_{1}\cdots A_{M}X_{1}\cdots X_{M}} = \sum_{a_{1},\ldots,a_{M},x_{1},\ldots,x_{M}}q(x_{1},\ldots,x_{M})p(a_{1},\ldots,a_{M}\vert x_{1},\ldots,x_{M})\\
                [a_{1},\ldots,a_{M},x_{1},\ldots,x_{M}]_{A_{1} \cdots A_{M}X_{1} \cdots X_{M}}\otimes\rho_{E}^{a_{1},\ldots,a_{M},x_{1},\ldots,x_{M}}.
            \end{multline}        
        \end{definition}
        
        We now derive a chain rule for the quantity $\widetilde{I}(A_{1,1}A_{1,2};\cdots;A_{i,1}A_{i,2};\cdots;A_{M,1}A_{M,2}\vert E)$ similar to that in Theorem~\ref{theorem_2}. In doing so, we generalize \eqref{crulemult} to conditional dual total correlation and every finite $M$, such that we can prove additivity and other relevant properties of dual multipartite (quantum) intrinsic non-locality. 
    
        \begin{theorem}\label{theorem:dual-total}
            For every multipartite state $\rho_{A_{1,1}A_{1,2}\cdots A_{i,1}A_{i,2}\cdots A_{M,1}A_{M,2}E}$, the following equality holds:
            \begin{multline}\label{crulemult:dual-total}
                \widetilde{I}(A_{1,1}A_{1,2};\cdots;A_{i,1}A_{i,2};\cdots;A_{M,1}A_{M,2}\vert E) = 
                \widetilde{I}(A_{1,2};\cdots;A_{M,2}\vert E) \\
                + \widetilde{I}(A_{1,1};\cdots;A_{M,1}\vert EA_{[M],2})+\sum_{i=1}^{M} I(A_{[M]\backslash \{i\},1};A_{i,2}\vert EA_{[M]\backslash \{i\},2}).
            \end{multline}
        \end{theorem}
        \begin{proof}
            By applying definitions and the chain rule for conditional entropy, we find that
            \begin{align}
                & \widetilde{I}(A_{1,1}A_{1,2};A_{2,1}A_{2,2};\cdots;A_{M,1}A_{M,2}\vert E)\nonumber\\
                & =\sum_{i=1}^{M}H(A_{[M]\backslash \{i\},1}A_{[M]\backslash \{i\},2}\vert E)-(m-1)H(A_{1,1}A_{1,2}A_{2,1}A_{2,2}\cdots
                    A_{M,1}A_{M,2}\vert E)\\
                & =\sum_{i=1}^{M}\left[  H(A_{[M]\backslash \{i\},2}\vert E)+H(A_{[M]\backslash \{i\},1}\vert EA_{[M]\backslash \{i\},2})\right]  \nonumber\\
                & \qquad-(m-1)\left[  H(A_{1,2}A_{2,2}\cdots A_{M,2}\vert E)-H(A_{1,1}A_{2,1}\cdots
                A_{M,1}\vert EA_{1,2}A_{2,2}\cdots A_{M,2})\right]  \\
                & =\widetilde{I}(A_{1,2};A_{2,2};\cdots;A_{M,2}\vert E) \notag\\
                &\qquad\qquad +\sum_{i=1}^{M}H(A_{[M]\backslash \{i\},1}\vert EA_{[M]\backslash \{i\},2}
                )-(m-1)H(A_{1,1}A_{2,1}\cdots A_{M,1}\vert EA_{[M],2}).
            \end{align}
            Continuing, we find that
            \begin{align}
                &\sum_{i=1}^{M}H(A_{[M]\backslash \{i\},1}\vert EA_{[M]\backslash \{i\},2}
                )-(m-1)H(A_{1,1}A_{2,1}\cdots A_{M,1}\vert EA_{[M],2})\nonumber\\
                & =\sum_{i=1}^{M} \left[H(A_{[M]\backslash \{i\},1}\vert EA_{[M]\backslash \{i\},2})-H(A_{[M]\backslash \{i\},1}\vert EA_{\left[  M\right],2})+H(A_{[M]\backslash \{i\},1}\vert EA_{\left[  M\right]  ,2})\right]\nonumber\\
                & \qquad-(m-1)H(A_{1,1}A_{2,1}\cdots A_{M,1}\vert EA_{[M],2})\\
                & =\sum_{i=1}^{M}I(A_{[M]\backslash \{i\},1};A_{i,2}\vert EA_{[M]\backslash \{i\},2})+\widetilde{I}(A_{1,1};A_{2,1};\cdots;A_{M,1}\vert EA_{[M],2}).
            \end{align}
            This concludes the proof.
        \end{proof}
        For the particular case of $M=3$, the expression in \eqref{crulemult:dual-total} reduces to
        \begin{multline}
            \widetilde{I}(A_{1}A_{2};B_{1}B_{2};C_{1}C_{2}\vert E) =
                \widetilde{I}(A_{2};B_{2};C_{2}\vert E) + \widetilde{I}(A_{1};B_{1};C_{1}\vert EA_{2}B_{2}C_{2})\\+ I(B_{1}C_{1};A_{2}\vert EB_{2}C_{2}) +I(A_{1}C_{1};B_{2}\vert EA_{2}C_{2}) +I(B_{1}A_{1};C_{2}\vert EB_{2}A_{2}).
        \end{multline}
        One can use the above equation to establish additivity of dual multipartite intrinsic non-locality for the tripartite case. Then, we can extend the arguments presented for the multipartite intrinsic non-locality to obtain additivity, convexity, and monotonicity under LOCR for dual multipartite intrinsic non-locality and dual multipartite quantum intrinsic non-locality, primarily due to the structure of \eqref{crulemult:dual-total} producing similar terms for every finite~$M$.

\section{Device-Independent Conference Key Agreement Capacity}
    
    \label{sec:di-cka-capacity}
    
    In this section, we define a general form of a tripartite device-independent conference key agreement protocol and its associated capacity. We shall then upper bound this capacity using tripartite intrinsic non-locality. Here, we show details of the definition for the case in which the eavesdropper possesses a no-signaling extension of the underlying correlation, and then we remark how the definition can be modified to the case in which the eavesdropper is restricted by quantum mechanics.
    
    Let $n\in\mathbb{Z}^{+}$, $R \geq 0$, and $\varepsilon\in\left[0,1\right]$. Let $p(a,b,c\vert x,y,z)$ be the correlation of the device shared by Alice, Bob, and Charlie. We define an $(n,R, \varepsilon)$ device-independent conference-key-agreement protocol as follows:

    \begin{itemize}\label{}
        \item Alice, Bob, and Charlie generate the input sequences $x^{n}$, $y^{n}$, and $z^{n}$ to their devices according to the probability distribution $q_{X^{n}Y^{n}Z^{n}}(x^{n},y^{n},z^{n})$. The device is used $n$ times, and the distribution $q_{X^{n}Y^{n}Z^{n}}(x^{n},y^{n},z^{n})$ is independent of the eavesdropper. For round $j \in\{1,\ldots,n\}$, Alice inputs $x_{j}$ and obtains the output $a_{j}$; Bob inputs $y_{j}$ and obtains the output $b_{j}$; Charlie inputs $z_{j}$ and obtains the output $c_{j}$. The distribution for the inputs and outputs can be  embedded in the state $\sigma_{A^{n}B^{n}C^{n}X^{n}Y^{n}Z^{n}}$, defined as
        \begin{multline}
            \sigma_{A^{n}B^{n}C^{n}X^{n}Y^{n}Z^{n}} = \sum_{a^{n},b^{n},c^{n},x^{n},y^{n},z^{n}}q_{X^{n}Y^{n}Z^{n}}(x^{n},y^{n},z^{n})p^{n}(a^{n},b^{n},c^{n}\vert x^{n},y^{n},z^{n})\\ \times\ketbra*{a^{n}b^{n}c^{n}x^{n}y^{n}z^{n}}{a^{n}b^{n}c^{n}x^{n}y^{n}z^{n}}_{A^nB^nC^nX^nY^nZ^n},
        \end{multline}
        where $p^{n}(a^{n},b^{n},c^{n}\vert x^{n},y^{n},z^{n})$ is the $n$-fold independent and identically distributed extension of $p(a,b,c\vert x,y,z)$. The joint state held by Alice, Bob, Charlie, and Eve is an arbitrary no-signaling extension  $\sigma_{A^{n}B^{n}C^{n}X^{n}Y^{n}Z^{n}E}$ of  $\sigma_{A^{n}B^{n}C^{n}X^{n}Y^{n}Z^{n}}$, as defined in \eqref{nosig_ext}.
        
        \item Alice performs a local channel $\mathcal{L}^{A}_{A^{n}\rightarrow M_{A}C_{A}}$, with $C_{A}$ denoting a classical register that is publicly communicated from Alice to Bob and Charlie, and $M_{A}$ denotes a classical local memory register that is not used for public communication. The register $\bar{C}_{A}$ is a classical register held by Eve, which is a copy of $C_{A}$. Similarly, Bob performs a local channel $\mathcal{L}^{B}_{B^{n}\rightarrow M_{B}C_{B}}$, with $C_{B}$ denoting the classical register that is publicly communicated from Bob to Alice and Charlie, and $M_{B}$ denotes a classical local memory register that is not used for public communication. The register $\bar{C}_{B}$ is a classical register held by Eve, which is a copy of $C_{B}$. Charlie performs a local channel $\mathcal{L}^{C}_{C^{n}\rightarrow M_{C}C_{C}}$, with $C_{C}$ denoting the classical register that is publicly communicated from Charlie to Bob and Alice, and $M_{C}$ denotes a classical local memory register, which is not used for public communication. The register $\bar{C}_{C}$ is a classical register held by Eve, which is a copy of $C_{C}$. The registers $C_A$, $C_B$, and $C_C$ (public communication) are used for parameter estimation. If the parameters are found to be outside of a predetermined range, the protocol is aborted and no secret key is agreed upon.
        
        \item Alice then performs the decoding channel $\mathcal{D}^{A}_{M_{A}C_{A}C_{B}C_{C}\rightarrow L_{A}}$ to obtain her final key system~$L_A$. Bob performs the decoding channel $\mathcal{D}^{B}_{M_{B}C_{A}C_{B}C_{C}\rightarrow L_{B}}$ to obtain his final key system~$L_B$. Charlie performs the decoding channel $\mathcal{D}^{C}_{M_{C}C_{A}C_{B}C_{C}\rightarrow L_{C}}$ to obtain his final key system~$L_C$. This protocol yields a state $\omega_{L_{A}L_{B}L_{C}EX^{n}Y^{n}Z^{n}\bar{C}_{A}\bar{C}_{B}\bar{C}_{C}}$ that satisfies
        \begin{equation}\label{cont1}
            \frac{1}{2}\left\Vert \Phi_{L_{A}L_{B}L_{C}EX^{n}Y^{n}Z^{n}\bar{C}_{A}\bar{C}_{B}\bar{C}_{C}}-\omega_{L_{A}L_{B}L_{C}EX^{n}Y^{n}Z^{n}\bar{C}_{A}\bar{C}_{B}\bar{C}_{C}} \right\Vert_{1} \leq \varepsilon,
        \end{equation}
        where
        \begin{equation}\label{state:ideal}
            \Phi_{L_{A}L_{B}L_{C}EX^{n}Y^{n}Z^{n}\bar{C}_{A}\bar{C}_{B}\bar{C}_{C}} = 2^{-nR} \sum_{l=1}^{2^{nR}}\ketbra*{l}{l}_{L_{A}}\otimes\ketbra*{l}{l}_{L_{B}}\otimes\ketbra*{l}{l}_{L_{C}}\otimes\omega_{EX^{n}Y^{n}Z^{n}\bar{C}_{A}\bar{C}_{B}\bar{C}_{C}}.
        \end{equation}
    \end{itemize}

    \begin{figure}[t]
        \centering\includegraphics[width=\textwidth]{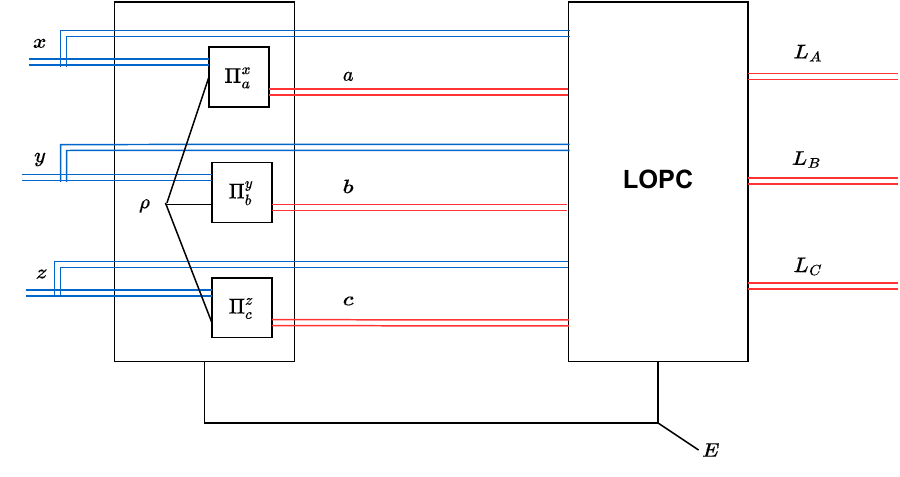}
        \caption{General schematic for device-independent conference key agreement. The POVMs $\{\Pi^{(x)}_{a}\}_{a}$, $\{\Pi^{(y)}_{b}\}_{b}$, and $\{\Pi^{(z)}_{c}\}_{c}$ are  available to Alice, Bob, and Charlie, respectively. The eavesdropper is in possession of the quantum and classical information in system $E$. LOPC stands for local operations and public communication and  is used by Alice, Bob, and Charlie to distill the final conference key.}
        \label{fig:pic}
    \end{figure}

    A general protocol of the above form is depicted in Figure~\ref{fig:pic}.
    A rate $R$ is achievable for a device characterized by a correlation $p$ if there exists an $(n,R-\delta, \varepsilon)$ device-independent conference key agreement protocol for all $\varepsilon \in (0, 1]$, $\delta > 0$, and sufficiently large $n$. The maximum achievable rate is denoted by $\operatorname{DI}(p)$ and is called the DI conference key agreement capacity.
    
    These definitions can easily be modified to the case in which the eavesdropper is restricted by quantum mechanics. The main modification is that the underlying correlation is a quantum correlation and the eavesdropper is allowed to possess a quantum extension of it. We denote the resulting capacity by $\operatorname{DI}_Q(p)$.
    
    It is straightforward to generalize everything stated above to the case of a multipartite correlation $p(a_1, \ldots, a_M \vert  x_1, \ldots, x_M)$.
    	
    In \cite{ribeiro2018fully}, a lower bound on conference key agreement rate was established for a particular protocol. In this work, we are trying to address a different question that can be answered regardless of any particular choice of protocol. We are concerned with the no-signaling or quantum correlations that characterize the devices used for device-independent conference key agreement.
    The question we want to answer is as follows: \textit{given a correlation $p(a,b,c\vert x,y,z)$, produced by a device, what is a non-trivial upper bound on the conference key agreement rate that can be extracted from this device with any possible protocol?}
    
    We answer this question for independent and identically distributed (i.i.d.) devices, which means, in each round of the protocol, the device is characterized by the correlation $p(a,b,c\vert x,y,z)$. The inputs within each round of the protocol can be correlated but not across rounds. This i.i.d.~assumption is not a drawback as we are interested in calculating \textit{upper} bounds on conference key agreement rates: if we show that a correlation can certify no more than a certain limit of key rate against an eavesdropper restricted to i.i.d.~attacks, then the correlation certainly cannot certify more than this limit against an eavesdropper without such a restriction. 

\subsection{Upper Bound on DI Conference Key Agreement Capacity}\label{subsec:tri-ups}

    Now, we prove that tripartite intrinsic non-locality is indeed an upper bound on the DI conference key agreement capacity. 
    \begin{theorem}\label{theorem_1}
        The tripartite intrinsic non-locality $N(A;B;C)_{p}$ is an upper bound on the device-independent conference key agreement capacity of a device characterized by the no-signaling correlation $p(a,b,c\vert x,y,z)$ and sharing no-signaling correlations with an eavesdropper:
        \begin{align}
            \operatorname{DI}(p) \leq N(A;B;C)_{p}.
        \end{align}
    \end{theorem}
    
    \begin{proof}
        The states $\Phi$, $\omega$, and $\sigma$ are given in the definition of device-independent conference key agreement in Section~\ref{sec:di-cka-capacity}. Using \eqref{cont1} and \eqref{cont}, we find that
        \begin{align}
            2nR &= I(L_{A};L_{B};L_{C}\vert EX^{n}Y^{n}Z^{n}\bar{C}_{A}\bar{C}_{B}\bar{C}_{C})_{\Phi}\\
            &\leq I(L_{A};L_{B};L_{C}\vert EX^{n}Y^{n}Z^{n}\bar{C}_{A}\bar{C}_{B}\bar{C}_{C})_{\omega} + \tilde{\varepsilon}.
        \end{align}
        where $\tilde{\varepsilon} = 4\varepsilon nR + 3g(\varepsilon)$ and $g(\varepsilon)$ is defined in \eqref{g}. Using data processing of conditional total correlation for $L_{A}$, $L_{B}$, and $L_{C}$ under the local channels  $\mathcal{D}^{A}_{M_{A}C_{A}C_{B}C_{C}\rightarrow L_{A}}$, $\mathcal{D}^{B}_{M_{B}C_{A}C_{B}C_{C}\rightarrow L_{A}}$, and $\mathcal{D}^{C}_{M_{C}C_{A}C_{B}C_{C}\rightarrow L_{A}}$, we conclude that
        \begin{align}
            2nR &\leq I(L_{A};L_{B};L_{C}\vert EX^{n}Y^{n}Z^{n}\bar{C}_{A}\bar{C}_{B}\bar{C}_{C})_{\omega} + \tilde{\varepsilon}\\
            &\leq I(M_{A}C_{A}C_{B}C_{C};M_{B}C_{A}C_{B}C_{C};M_{C}C_{A}C_{B}C_{C}\vert EX^{n}Y^{n}Z^{n}\bar{C}_{A}\bar{C}_{B}\bar{C}_{C})_{\omega} + \tilde{\varepsilon}.
        \end{align}
        Now, since $\bar{C}_B$ is a copy of $C_B$ and $\bar{C}_C$ is a copy of $C_C$, we conclude that
        \begin{equation}
            H(M_{A}C_{A}C_{B}C_{C}\vert EX^{n}Y^{n}Z^{n}\bar{C}_{A}\bar{C}_{B}\bar{C}_{C}) = H(M_{A}C_{A}\vert EX^{n}Y^{n}Z^{n}\bar{C}_{A}\bar{C}_{B}\bar{C}_{C}).
        \end{equation}
        A similar manipulation can be applied to $H(M_{B}C_{A}C_{B}C_{C}\vert EX^{n}Y^{n}Z^{n}\bar{C}_{A}\bar{C}_{B}\bar{C}_{C})$ and \\ $H(M_{C}C_{A}C_{B}C_{C}\vert EX^{n}Y^{n}Z^{n}\bar{C}_{A}\bar{C}_{B}\bar{C}_{C})$, giving us 
        \begin{align}
            2nR &\leq I(M_{A}C_{A}C_{B}C_{C};M_{B}C_{A}C_{B}C_{C};M_{C}C_{A}C_{B}C_{C}\vert EX^{n}Y^{n}Z^{n}\bar{C}_{A}\bar{C}_{B}\bar{C}_{C})_{\omega} + \tilde{\varepsilon}\nonumber\\
            &\leq I(M_{A}C_{A};M_{B}C_{B};M_{C}C_{C}\vert EX^{n}Y^{n}Z^{n}\bar{C}_{A}\bar{C}_{B}\bar{C}_{C})_{\omega} + \tilde{\varepsilon}.
        \end{align}
        Using \eqref{eqn} 
        and ignoring the negative terms that arise, we find that
        \begin{align}
            2nR & \leq I(M_{A}C_{A};M_{B}C_{B};M_{C}C_{C}\vert EX^{n}Y^{n}Z^{n}\bar{C}_{A}\bar{C}_{B}\bar{C}_{C})_{\omega} + \tilde{\varepsilon}\nonumber\\
            &\leq I(M_{A}C_{A}\bar{C}_{A};M_{B}C_{B}\bar{C}_{B};M_{C}C_{C}\bar{C}_{C}\vert EX^{n}Y^{n}Z^{n})_{\omega} + \tilde{\varepsilon}\nonumber\\
            &= I(M_{A}C_{A};M_{B}C_{B};M_{C}C_{C}\vert EX^{n}Y^{n}Z^{n})_{\omega} + \tilde{\varepsilon}.
        \end{align}
        Using data processing of conditional total correlation on $M_{A}C_{A}$, $M_{B}C_{B}$, and $M_{C}C_{C}$,
        \begin{align}\label{eqn2}
            2nR&\leq I(M_{A}C_{A};M_{B}C_{B};M_{C}C_{C}\vert EX^{n}Y^{n}Z^{n})_{\omega} + \tilde{\varepsilon}\nonumber\\
            &\leq I(A^{n};B^{n};C^{n}\vert EX^{n}Y^{n}Z^{n})_{\sigma} + \tilde{\varepsilon}.
        \end{align}
        Using the fact that the no-signaling extension applied in the protocol in Section~\ref{sec:di-cka-capacity} is arbitrary,
        \begin{align}
            2nR& \leq \inf_{\text{ext.}}I(A^{n};B^{n};C^{n}\vert EX^{n}Y^{n}Z^{n})_{\sigma} + \tilde{\varepsilon}
        \end{align}
        Using $\tilde{\varepsilon} = 4\varepsilon nR + 3g(\varepsilon)$,
        \begin{align}
            2(1-2\varepsilon)nR& \leq \inf_{\text{ext.}} I(A^{n};B^{n};C^{n}\vert EX^{n}Y^{n}Z^{n})_{\sigma} + 3g(\varepsilon).
        \end{align}
        Taking the supremum over all input distributions,
        \begin{align}
            2(1-2\varepsilon)nR& \leq \sup_{q}\inf_{\text{ext.}} I(A^{n};B^{n};C^{n}\vert EX^{n}Y^{n}Z^{n})_{\sigma} + 3g(\varepsilon).
        \end{align}
        Using additivity (see Theorem~\ref{theorem:additivity}),
        \begin{align}
            2(1-2\varepsilon)nR& \leq \sup_{q}\inf_{\text{ext.}} I(A^{n};B^{n};C^{n}\vert EX^{n}Y^{n}Z^{n})_{\rho} + 3g(\varepsilon)\\
            &=  n\cdot \sup_{q}\inf_{\text{ext.}} I(A;B;C\vert EXYZ)_{\rho} + 3g(\varepsilon)\\
            \implies 2(1-2\varepsilon)R& \leq \sup_{q}\inf_{\text{ext.}} I(A;B;C\vert EXYZ)_{\rho} + \frac{3}{n}g(\varepsilon).
        \end{align}
        Taking the limit $n\rightarrow\infty$ and then $\varepsilon\rightarrow 0$, we conclude that
        \begin{align}
            \operatorname{DI}(p) \leq N(A;B;C).
        \end{align}
        This concludes the proof.
    \end{proof}
    
    Using similar techniques and taking appropriate quantum extensions establishes the following:
    \begin{theorem}\label{theorem_1q}
        The quantum tripartite intrinsic non-locality $N_Q(A;B;C)_{p}$ is an upper bound on the device-independent conference key agreement capacity of a device characterized by the quantum correlation $p(a,b,c\vert x,y,z)$ and sharing quantum correlations with an eavesdropper:
        \begin{align}
            \operatorname{DI}_{Q}(p) \leq N_{Q}(A;B;C)_{p}.
        \end{align}
    \end{theorem}

    All the steps (i.e., data processing and additivity) in the proof of Theorem~\ref{theorem_1} can be easily extended to apply to multipartite intrinsic non-locality, dual multipartite intrinsic non-locality, and their respective quantum counterparts. This leads to the following theorems:
    \begin{theorem}\label{theorem_4}
        The multipartite intrinsic non-locality $N(A_{1};\cdots;A_{M})_{p}$ is an upper bound on the device-independent conference key agreement capacity of a device characterized by a no-signaling correlation $p(a_{1},\ldots,a_{M}\vert x_{1},\ldots,x_{M})$ and sharing no-signaling correlations with an eavesdropper:
        \begin{align}
            \operatorname{DI}(p) \leq N(A_{1};\cdots;A_{M})_{p}.
        \end{align}
    \end{theorem}

    \begin{theorem}\label{theorem_4-q} The
        multipartite quantum intrinsic non-locality $N_{Q}(A_{1};\cdots;A_{M})_{p}$ is an upper bound on the device-independent conference key agreement capacity of a device characterized  by a quantum correlation $p(a_{1},\ldots,a_{M}\vert x_{1},\ldots,x_{M})$ and sharing quantum correlations with an eavesdropper:
        \begin{align}
            \operatorname{DI}_{Q}(p) \leq N_{Q}(A_{1};\cdots;A_{M})_{p}.
        \end{align}
    \end{theorem}

    \begin{theorem}\label{theorem_5}
        Dual multipartite intrinsic non-locality $\widetilde{N}(A_{1};\cdots;A_{M})_{p}$ is an upper bound on the device-independent conference key agreement capacity of a device characterized by a no-signaling correlation $p(a_{1},\ldots,a_{M}\vert x_{1},\ldots,x_{M})$ and sharing no-signaling correlations with an eavesdropper:
        \begin{align}
            \operatorname{DI}(p) \leq \widetilde{N}(A_{1};\cdots;A_{M})_{p}.
        \end{align}
    \end{theorem}

    \begin{theorem}\label{theorem_5-q} Dual
        multipartite quantum intrinsic non-locality $\widetilde{N}_{Q}(A_{1};\cdots;A_{M})_{p}$ is an upper bound on the device-independent conference key agreement capacity of a device characterized  by a quantum correlation $p(a_{1},\ldots,a_{M}\vert x_{1},\ldots,x_{M})$ and sharing quantum correlations with an eavesdropper:
        \begin{align}
            \operatorname{DI}_{Q}(p) \leq \widetilde{N}_{Q}(A_{1};\cdots;A_{M})_{p}.
        \end{align}
    \end{theorem}

\section{Evaluating Quantum Tripartite Intrinsic Non-Locality}

\label{sec:ex-ex}

    In this section, we evaluate quantum tripartite intrinsic non-locality for various examples. While evaluating the quantum tripartite intrinsic non-locality, we should consider the actions of an eavesdropper, who is in possession of an extension of the underlying quantum state shared by Alice, Bob, and Charlie. We note here that all source files needed to generate the plots in this section are available with the arXiv posting of this paper.
    
    An eavesdropper, Eve, of a DIQKD protocol is allowed access to the quantum extension system of the state shared between Alice and Bob prior to public communication of measurement settings. 
    Eve is also assumed to be in possession of copies of all classical communication exchanged by Alice and Bob, as well as all local hidden variables that can be attributed to the correlations that Alice and Bob share. We also assume that the state and black boxes received by Alice and Bob are in fact supplied by Eve herself. 
    
    For DI conference key agreement protocols, we assume that Eve has access to all the same quantum and classical information as in DIQKD but sourced from all the participants of the DI conference key agreement protocol. Eve can then use this collected information to reduce the key agreement rate. Any procedure employed by Eve to reduce the key agreement rate is known as an attack.
    
     The first attack that we consider is a modification of the attack for DIQKD used in~\cite{kaur2020fundamental}, which was helpful for calculating an upper bound on quantum intrinsic non-locality. We use the RMW18 Protocol \cite[]{ribeiro2018fully} 
     for all further calculations. First, suppose that the underlying state is as follows:
     \begin{align}\label{state1}
        \rho_{\tilde{A}\tilde{B}\tilde{C}}=
        (1-p)\ketbra*{\text{GHZ}}{\text{GHZ}}_{\tilde{A}\tilde{B}\tilde{C}} + p\frac{\mathbb{I}_{\tilde{A}\tilde{B}\tilde{C}}}{8}  ,
    \end{align}
    where $\ket*{\text{GHZ}}=\left(\ket*{000}+\ket*{111}\right)/\sqrt{2}$.  Alice's measurement choice $x=0$ corresponds to~$\sigma_{Z}$, and $x=1$ corresponds to $\sigma_{X}$. Bob's measurement choice $y=0$ corresponds to~$(\sigma_{Z}-\sigma_{X})/\sqrt{2}$, the choice $y=1$ corresponds to $(\sigma_{Z}+\sigma_{X})/\sqrt{2}$, and the choice $y=2$ corresponds to $\sigma_{Z}$. Charlie's measurement choices are $\sigma_{Z}$ when $z=0$ and $\sigma_{X}$ when $z=1$. This leads to a quantum correlation $q(a,b,c|x,y,z)$.
    
    \begin{figure}[t]
        \centering\includegraphics[width=0.8\textwidth]{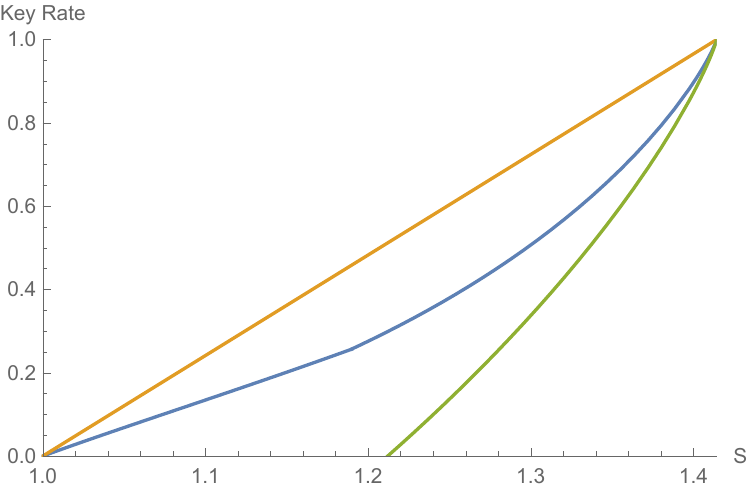}
        \caption{Key rate versus parity-CHSH violation $S$. The orange line is an upper bound on quantum tripartite intrinsic non-locality computed for the attack described in \eqref{attack1}, the blue line is an upper bound on quantum tripartite intrinsic non-locality for the correlation parameterized by $S$ using a multipartite generalization of the attack in \cite{pironio2009device,arnon2021upper}, and the green solid line is the lower bound for the state in \eqref{state1} calculated from \cite{ribeiro2018fully}. }
        \label{fig:plot1}
    \end{figure}
    
    Using the Bell inequality corresponding to the parity-CHSH game \cite{ribeiro2018fully,holz2020genuine}, the parity-CHSH violation $S$ is as follows:\footnote{The calculations for all $S$, $p_{\text{win}}$, and plots are in the Mathematica files included with the arXiv posting of our paper.}
    \begin{align}
        S=\sqrt{2}(1-p).
    \end{align}
    We see that $\rho_{\tilde{A}\tilde{B}\tilde{C}}$  produces a local correlation when the parity-CHSH violation is less than or equal to one or, equivalently, when $p\geq1-1/\sqrt{2}$. Let $q_{S^{p}}(a,b,c\vert x,y,z)$ denote a quantum correlation with parity-CHSH violation $S^{p}$. For $\varepsilon\leq p \leq 1-\frac{1}{\sqrt{2}}$, we can think of the correlation $q_{S^{p}}(a,b,c\vert x,y,z)$  as a convex combination of $q_{S^{\varepsilon}}(a,b,c\vert x,y,z)$, which is non-local, and $q_{S^{1-\frac{1}{\sqrt{2}}}}(a,b,c\vert x,y,z)$, which is local, in the following fashion:
    \begin{align}
        q_{S^{p}}(a,b,c\vert x,y,z) = (1-\alpha(\varepsilon))q_{S^{\varepsilon}}(a,b,c\vert x,y,z)+\alpha(\varepsilon)q_{S^{1-\frac{1}{\sqrt{2}}}}(a,b,c\vert x,y,z),
    \end{align}
    where 
    \begin{align}
        \alpha(\varepsilon)=\frac{p-\varepsilon}{1-\frac{1}{\sqrt{2}}-\varepsilon}.
    \end{align}
    For local correlations, quantum tripartite intrinsic non-locality is equal to zero. Hence, using Theorem~\ref{theorem:tri-conv-q}, we conclude that
    \begin{align}
        N_{Q}(A;B;C)_{q_{S^{p}}}\leq (1-\alpha(\varepsilon))N_{Q}(A;B;C)_{q_{S^{\varepsilon}}}.
    \end{align}
    By considering the trivial extension for $q_{S^{p}}(a,b,c\vert x,y,z)$, we obtain
    \begin{align}\label{attack1}
        N_{Q}(A;B;C)_{q_{S^{p}}}\leq \min_{0\leq\varepsilon\leq p}\sup_{q(x,y,z)}(1-\alpha(\varepsilon))I(A;B;C)_{q_{S^{\varepsilon}}}.
    \end{align}
    The lower bound is calculated using the probability of winning the parity-CHSH game, given by
    \begin{align}
        p_{\text{win}} = \frac{1}{2}+\frac{(1-p)}{2\sqrt{2}}.
    \end{align}
    We then plot this quantity against the parity-CHSH violation $S$ in Figure~\ref{fig:plot1}.
        
    The second attack on the RMW18 Protocol \cite[]{ribeiro2018fully}
    that we consider is a multipartite generalization of the attack on DIQKD first proposed in \cite{pironio2009device}, in the context of a lower bound. It has also been used in \cite{arnon2021upper} for evaluating an upper bound on DIQKD. It can be thought of as a particular way of achieving a desired parity-CHSH violation $S$ and quantum bit error rate (QBER) $Q$. In the multipartite generalization, we consider the following state:
    \begin{equation}
    \frac{1-C}{2} (Z_{\tilde{A}} \otimes Z_{\tilde{B}} \otimes Z_{\tilde{C}})( |\text{GHZ}\rangle\!\langle\text{GHZ}|_{\tilde{A}\tilde{B}\tilde{C}}) (Z_{\tilde{A}} \otimes Z_{\tilde{B}} \otimes Z_{\tilde{C}})  + 
        \frac{1+C}{2} |\text{GHZ}\rangle\!\langle\text{GHZ}|_{\tilde{A}\tilde{B}\tilde{C}},
    \end{equation}
    which results from the action of collective dephasing on the GHZ state, and which is purified by the following state vector:
    \begin{align}
        \sqrt{\frac{1-C}{2}}\left(\frac{\ket*{000}-\ket*{111}}{\sqrt{2}}\right)_{\tilde{A}\tilde{B}\tilde{C}}\otimes\ket*{0}_{E} + \sqrt{\frac{1+C}{2}}\left(\frac{\ket*{000}+\ket*{111}}{\sqrt{2}}\right)_{\tilde{A}\tilde{B}\tilde{C}}\otimes\ket*{1}_{E}.
    \end{align}
    
    Alice's measurement choice $x=0$ corresponds to $\sigma_{Z}$, and $x=1$ corresponds to $\sigma_{X}$. Bob's measurement choice $y=0$ corresponds to $(\sigma_{Z}+C\sigma_{X})/\sqrt{1+C^{2}}$, the choice $y=1$ corresponds to $(\sigma_{Z}-C\sigma_{X})/\sqrt{1+C^{2}}$, and $y=2$ corresponds to $\sigma_{Z}$. Charlie's measurement choices are $\sigma_{Z}$ when $z=0$ and $\sigma_{Z}$ when $z=1$. The parity-CHSH violation $S$ is given by $S=\sqrt{1+C^{2}}$. To generate key, Alice and Charlie measure $\sigma_{Z}$ and Bob, with probability $1-2Q$, measures $\sigma_{Z}$ and, with probability $2Q$, assigns a random bit. This gives us a QBER of $Q$. The post-measurement state is as follows:
    \begin{align}\label{state3}
        \frac{1-Q}{2}\left(\ketbra*{000}{000}\otimes\rho_{E}^{+}+\ketbra*{111}{111}\otimes\rho_{E}^{-}\right)+\frac{Q}{2}\left(\ketbra*{001}{001}\otimes\rho_{E}^{+}+\ketbra*{110}{110}\otimes\rho_{E}^{-}\right),
    \end{align}
    where
    \begin{align}
        \rho_{E}^{\pm}=\frac{1}{2}
        \begin{pmatrix}
            1+C & \pm\sqrt{1-C^{2}}\\
            \pm\sqrt{1-C^{2}} & 1-C 
        \end{pmatrix}.
    \end{align}
     Note that for the state in \eqref{state1}, the parity-CHSH violation $S$ and QBER $Q$ are related as follows: $Q=\frac{1}{2}(1-\frac{S}{\sqrt{2}})$. After we apply this relation between $S$ and $Q$, we get a correlation that is parameterized by $S$. We then calculate an upper bound on quantum tripartite intrinsic non-locality as a function of $S$ and plot it versus $S$ in Figure~\ref{fig:plot1}. It is important to note that this parameterized correlation is not convex in the parameter $S$, as required by \eqref{eqn:convex}; so if such a curve is not convex to begin with, Theorem~\ref{theorem:tri-conv-q} cannot be invoked to produce a lower, convex curve that is also an upper bound on the quantum tripartite non-locality for the parameterized correlations. We will encounter such a non-convex upper bound curve in Figure~\ref{fig:plot3} of the next section.
    
    \begin{figure}[t]
        \centering\includegraphics[width=0.8\textwidth]{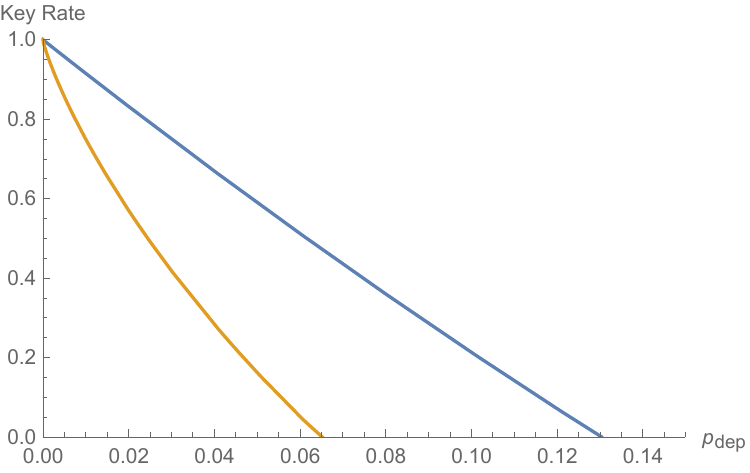}
        \caption{The blue line is the plot of tripartite intrinsic non-locality as function of $p_{\text{dep}}$ for the state $\mathcal{D}^{\otimes 3}(\ketbra*{\text{GHZ}}{\text{GHZ}})$ using the attack leading to \eqref{attack1}. The gold line indicates the lower bound calculated from \cite{ribeiro2018fully}. }
        \label{fig:plot2}
    \end{figure}
    
    A common qubit noise model is the depolarizing channel, described as
    \begin{align}
        \mathcal{D}(\rho) \coloneqq  (1-p_{\text{dep}})\rho+p_{\text{dep}}\frac{\mathbb{I}}{2}.
    \end{align}
    We can then consider a more realistic noise model given by  $\rho_{\tilde{A}\tilde{B}\tilde{C}}= \mathcal{D}^{\otimes 3}(\ketbra*{\text{GHZ}}{\text{GHZ}})$. For this state, we can consider the attack leading to \eqref{attack1} using the parity-CHSH violation~$S$, given by
     \begin{align}
        S = \frac{(1-p_{\text{dep}})^{3}}{\sqrt{2}}+\frac{(1-p_{\text{dep}})^{2}}{\sqrt{2}}.
    \end{align}
    The lower bound from \cite{ribeiro2018fully} is calculated using the probability of winning the parity-CHSH game, given by
    \begin{align}
        p_{\text{win}} = \frac{1}{2}+\frac{(1-p_{\text{dep}})^{3}}{2\sqrt{2}}+\frac{p(1-p_{\text{dep}})^{2}}{4\sqrt{2}}.
    \end{align}
    We plot quantum tripartite intrinsic non-locality against $p_{\text{dep}}$ in Figure~\ref{fig:plot2}. 
    
    Here we note that tripartite intrinsic non-locality based on dual total correlation provides the exact same upper bounds when calculated using the attack in \eqref{attack1}. For the other examples we have studied, tripartite intrinsic non-locality based on conditional dual total correlation gives worse upper bounds than multipartite intrinsic non-locality based on conditional total correlation.

\section{Upper Bound Evaluation for Experimental DIQKD}\label{sec:UpBo-DIQKD}

    The primary focus of this paper has been DI conference key agreement. An experimental implementation of DI conference key agreement beyond two parties is still in its infancy as of the writing of this paper, though recent progress in multi-party quantum nonlocality experiments is underway \cite{huang2022experimental}. The main interest in conference key agreement is as a way to quickly establish secret key among several parties, possibly linked by a quantum network. Another, more currently accessible, method to achieve this is the development and use of highly efficient DIQKD protocols between pairs of individual users, who can then use these protocols to distribute a single key to all parties. While this method is not as efficient as a genuine three party approach, \cite{epping2016large}, exploring advances in the bipartite scenario of DI conference key agreement or device-independent quantum key distribution (DIQKD) will be relevant to DI conference key agreement.
    
    Recently, there have been experimental works implementing DIQKD \cite{zhang2021experimental,nadlinger2021deviceindependent,liu2021highspeed}. The protocol used by \cite{zhang2021experimental} is of particular interest to us as it uses two key generation rounds, unlike most others which have just one key generation setting. We calculate an upper bound for this bipartite protocol.
    An upper bound on the DIQKD rate is given by quantum intrinsic non-locality, as shown in \cite{kaur2020fundamental}. For the experimental protocol in \cite{zhang2021experimental}, we consider the attack proposed by \cite{arnon2021upper} and calculate quantum intrinsic non-locality \cite{kaur2020fundamental}.
    
    \begin{figure}[t]
        \centering\includegraphics[width=0.8\textwidth]{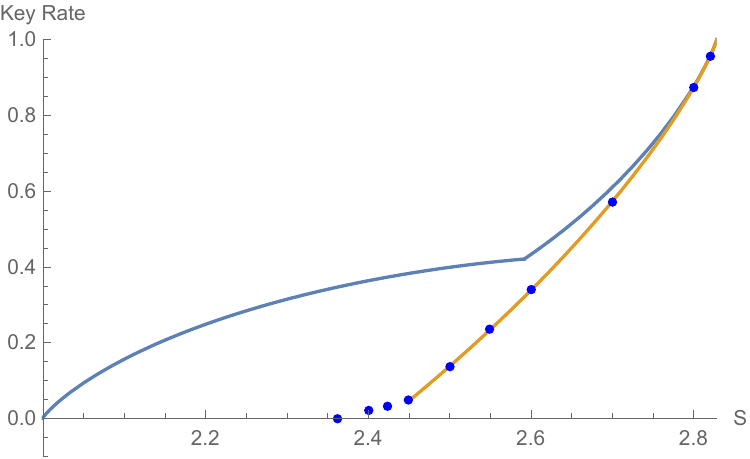}
        \caption{The blue line is an upper bound on quantum intrinsic non-locality versus CHSH violation, calculated for the correlation parameterized by $S$, which is described in the attack \cite{arnon2021upper} for the protocol proposed in \cite{schwonnek2021device}. The blue dots indicate the lower bound of the protocol proposed in \cite{schwonnek2021device}. The yellow line is the lower bound from \cite{pironio2009device}, which coincides with the lower bound of the protocol proposed in~\cite{schwonnek2021device} for certain values of $S$.}
        \label{fig:plot3}
    \end{figure}
    
    In the protocol used in \cite{zhang2021experimental}, Bob has two key generation settings that are picked randomly with equal probability. This experimental protocol is based on the protocol proposed by \cite{schwonnek2021device}. We set Alice's measurement choice $x=0$ to correspond to $\sigma_{Z}$, and $x=1$ corresponds to $\sigma_{X}$. We set Bob's measurement choice $y=0$ to correspond to $(\sigma_{Z}+C\sigma_{X})/\sqrt{1+C^{2}}$, the choice $y=1$ to correspond to $(\sigma_{Z}-C\sigma_{X})/\sqrt{1+C^{2}}$, and $y=2$ corresponds to $\sigma_{Z}$. Bob uses $y\in \{0,1\}$ for key generation, each with its own QBER. We calculate quantum intrinsic non-locality as function of $S$ using the attack described in \cite{arnon2021upper} while setting both QBERs to $Q=\frac{1}{2}(1-\frac{S}{2\sqrt{2}})$, where $S$ is the CHSH violation. This relation between QBER and CHSH violation holds for the Werner state. We then get a correlation that is parameterized by $S$. 
    Figure~\ref{fig:plot3} plots quantum intrinsic non-locality for this protocol versus CHSH violation.

\section{Conclusion}\label{sec:con-clues}
    
    In this paper, we defined multipartite intrinsic non-localities using conditonal total correlation and conditional dual total correlation, and we proved that these quantities are indeed additive and convex upper bounds on the DI conference key agreement capacity. These multipartite intrinsic non-localities are also monotone under local operations and common randomness. A key technical contribution is our derivation of the chain rule for conditional total correlation and conditional dual total correlation, which are applicable to all correlations and may be of independent interest beyond their applications to conference key agreement.

    For future work, we are interested in pursuing more novel DI conference key agreement protocols beyond the one presented in \cite{ribeiro2018fully}. Specifically, one can look for protocols that have more than one measurement setting in the key generation phase because such protocols require lower detector efficiency for DI quantum key distribution, as shown in \cite{PhysRevA.103.052436}. We could also investigate other Bell inequalities presented in \cite{holz2020genuine} in order to find better protocols. 
    
    One may also be interested in determining if either multipartite intrinsic non-locality is indeed a monotone of genuine multipartite Bell non-locality. It is also easy to see that multipartite intrinsic non-locality is equal to zero for correlations that can be described by a local hidden variable common to all parties involved. However, multipartite intrinsic non-locality is not known to be equal to zero for correlations that fail to be genuinely multipartite nonlocal as defined in \cite{bancal2013definitions}, such as (for instance) tripartite correlations that can be decomposed into a convex mixture of correlations that are each only bipartite nonlocal.

    We can also see from Figure~\ref{fig:plot1} that there is a significant gap between the upper  and lower bounds on tripartite DI conference key agreement, so that there is room for improvement. We also want to find new attacks specific to DI conference key agreement to improve the upper bound further and bring it closer to the lower bound. One can also look at convex combinations of various attacks on DI conference key agreement, as shown for DI~quantum key distribution in \cite{kaur2021upper}. Deriving a different multipartite intrinsic non-locality using another information quantity may also be of interest to improve the upper bound.

    Finally, we can also look at securing device-independent conference key agreement using just computational assumptions. There have already been attempts at securing DI quantum key distribution and self testing under computational assumptions based on the learning with errors problem \cite{metger2020deviceindependent,Metger2021selftestingofsingle}. It may be interesting to extend this analysis to the multipartite scenario of DI conference key agreement.

    \textit{Note Added}---We uploaded the first version of our preprint \cite{philip2021intrinsic} to the quant-ph arXiv concurrently with the first version of \cite{HWD21}, after being made aware of their independent work. Ref.~\cite{HWD21} has now been published as \cite{PhysRevA.105.022604}.
    
    \textit{Acknowledgements}---We acknowledge funding from Air Force Office of Scientific Research Award No.~FA9550-20-1-0067. We thank Ignatius W.~Primaatmaja and Charles C.-W.~Lim for their help in producing Figure~\ref{fig:plot3}. We also acknowledge helpful discussions with Soumyadip Patra.

\bibliography{ref}
\bibliographystyle{unsrt}

\appendix

    \section{Convexity}\label{subsec:tri-conv}

    Convexity of tripartite intrinsic non-locality is another important property because convex combinations of no-signaling correlations are also valid no-signaling correlations. This is also the case for quantum correlations \cite{doi:10.1063/1.527066}.
    
    \begin{theorem}[Convexity of TINL]\label{theorem:tri-conv}
        Let $t(a, b, c\vert x, y, z)$ and $r(a, b, c\vert x, y, z)$ be two no-signaling correlations, and let $\lambda \in [0, 1]$. Let $p(a, b, c\vert x, y, c)$ be a mixture of the two correlations, defined as 
        \begin{align}\label{eqn:convex}
            p(a, b, c\vert x, y, c) = \lambda t(a, b, c\vert x, y, z) + (1 - \lambda) r(a, b, c\vert x, y, z).
        \end{align}
        Then,
        \begin{align}
            N(A;B;C)_{p} \leq \lambda N(A;B;C)_{t} + (1 - \lambda) N(A;B;C)_{r}.
        \end{align}
    \end{theorem}
    
    \begin{proof}
        Consider the quantum embeddings of arbitrary no-signaling extensions of $t$, $r$, and~$p$:
        \begin{align}\label{ext_conv-1}
            \tau_{ABCEXYZ} &= \sum_{a,b,c,x,y,z} p(x,y,z)t(a,b,c\vert x,y,z)[abcxyz]_{ABCXYZ}\otimes\tau_{E}^{abcxyz},\\
            \gamma_{ABCEXYZ} &= \sum_{a,b,c,x,y,z} p(x,y,z)r(a,b,c\vert x,y,z)[abcxyz]_{ABCXYZ}\otimes\gamma_{E}^{abcxyz},
            \label{ext_conv-2}
        \end{align}
        and 
        \begin{align}\label{ext_conv-3}
            \zeta_{ABCEXYZ} &= \sum_{a,b,c,x,y,z} p(x,y,z)p(a,b,c\vert x,y,z)[abcxyz]_{ABCXYZ}\otimes\rho_{E}^{abcxyz}\nonumber\\
            \quad &= \sum_{a,b,c,x,y,z}  p(x,y,z)\{( \lambda)t(a,b,c\vert x,y,z) + (1 - \lambda)r(a,b,c\vert x,y,z) \}\nonumber\\
            &\qquad \qquad\quad \times [abcxyz]_{ABCXYZ}\otimes\rho_{E}^{abcxyz}.
        \end{align}
        A particular no-signaling extension of \eqref{ext_conv-3} is as follows:
        \begin{multline}
            \rho_{ABCEXYZ\Lambda} 
            = \sum_{a,b,c,x,y,z,\lambda}  p(x,y,z)\{( \lambda)t(a,b,c\vert x,y,z) [a,b,c,x,y,z]_{ABCXYZ}\otimes\tau_{E}^{abcxyz}\otimes[0]_{\Lambda} \\
             + (1 - \lambda)r(a,b,c\vert x,y,z) [abcxyz]_{ABCXYZ}\otimes\gamma_{E}^{abcxyz}\otimes[1]_{\Lambda}\}.
        \end{multline}
        Consider then
        \begin{align}
            &\inf_{\text{ext. in }\eqref{ext_conv-3}} I(A;B;C \vert EXYZ)_{\zeta} \notag\\
            & \leq I(A;B;C \vert EXYZ\Lambda)_{\rho} \\
            &= (\lambda)I(A;B;C \vert EXYZ)_{\tau} + (1 - \lambda)I(A;B;C \vert EXYZ)_{\gamma}\\
            &\leq (\lambda)\inf_{\text{ext. in }\eqref{ext_conv-1}}I(A;B;C \vert EXYZ)_{\tau} + (1 - \lambda)\inf_{\text{ext. in }\eqref{ext_conv-2}}I(A;B;C \vert EXYZ)_{\gamma}.
        \end{align}
        The first inequality holds because we picked a particular no-signaling extension. The second inequality holds due to the convexity of the individual terms in the definition of conditional total correlation. Since $\tau$ and $\gamma$ are arbitrary no-signaling extensions of $t$ and~$r$, and optimizing over arbitrary input probability distributions, we find that
        \begin{multline}
            \sup_{q}\inf_{\text{ext. in }\eqref{ext_conv-3}} I(A;B;C \vert EXYZ)_{p} \\
            \leq (\lambda)\sup_{q}\inf_{\text{ext. in }\eqref{ext_conv-1}}I(A;B;C \vert EXYZ)_{t} + (1 - \lambda)\sup_{q}\inf_{\text{ext. in }\eqref{ext_conv-2}}I(A;B;C \vert EXYZ)_{r}.
        \end{multline} 
        This concludes the proof.
    \end{proof}
    \begin{theorem}[Convexity of QTINL]\label{theorem:tri-conv-q}
        Let $t(a, b, c\vert x, y, z)$ and $r(a, b, c\vert x, y, z)$ be two quantum correlations, and let $\lambda \in [0, 1]$. Let $p(a, b, c\vert x, y, c)$ be a mixture of the two correlations, defined as 
        \begin{align}
            p(a, b, c\vert x, y, c) = \lambda t(a, b, c\vert x, y, z) + (1 - \lambda) r(a, b, c\vert x, y, z).
        \end{align}
        Then,
        \begin{align}
            N_{Q}(A;B;C)_{p} \leq \lambda N_{Q}(A;B;C)_{t} + (1 - \lambda) N_{Q}(A;B;C)_{r}.
        \end{align}
    \end{theorem}
    \begin{proof}
        Consider the following quantum extensions of $t$, $r$, and $p$:
        \begin{align}\label{ext_conv-1-q}
            \tau_{ABCEXYZ} & = \sum_{a,b,c,x,y,z} q(x,y,z)t(a,b,c\vert x,y,z)[abcxyz]_{ABCXYZ}\otimes\tau_{E}^{abcxyz},\\
        \label{ext_conv-2-q}
            \gamma_{ABCEXYZ} & = \sum_{a,b,c,x,y,z} q(x,y,z)r(a,b,c\vert x,y,z)[abcxyz]_{ABCXYZ}\otimes\gamma_{E}^{abcxyz},\\
        \label{ext_conv-3-q}
            \zeta_{ABCEXYZ} & = \sum_{a,b,c,x,y,z} q(x,y,z)p(a,b,c\vert x,y,z)[abcxyz]_{ABCXYZ}\otimes\rho_{E}^{abcxyz}.
        \end{align}
        Let $\tau_{\tilde{A}\tilde{B}\tilde{C}}$ be a quantum state that, along with the POVMs characterized by $\{\Pi^{(x)}_{a}\}_{a}$, $\{\Pi^{(y)}_{b}\}_{b}$, and  $\{\Pi^{(z)}_{c}\}_{c}$, yield the correlation $t(a,b,c\vert x,y,c)$. Let $\tau_{\tilde{A}\tilde{B}\tilde{C}E}$ be a quantum extension of $\tau_{\tilde{A}\tilde{B}\tilde{C}}$. Similarly, let $\gamma_{\tilde{A}\tilde{B}\tilde{C}}$ be a quantum state that, along with the POVMs characterized by $\{\Lambda^{(x)}_{a}\}_{a}$, $\{\Lambda^{(y)}_{b}\}_{b}$, and $\{\Lambda^{(z)}_{c}\}_{c}$, yield the correlation $r(a,b,c \vert x,y,z)$. Let $\gamma_{\tilde{A}\tilde{B}\tilde{C}E}$ be a quantum extension of $\gamma_{\tilde{A}\tilde{B}\tilde{C}}$. Then, a particular quantum state that realizes the correlation $p(a, b ,c\vert x, y,z)$ is the following:
        \begin{align}
        \rho_{\tilde{A}\tilde{B}\tilde{C}A^{\prime} B^{\prime}C^{\prime}} &=\lambda \tau_{\tilde{A}\tilde{B}\tilde{C}} \otimes\ketbra*{000}{000}_{A^{\prime} B^{\prime}C^{\prime}}+(1-\lambda) \gamma_{\tilde{A}\tilde{B}\tilde{C}} \otimes \ketbra*{111}{111}_{A^{\prime} B^{\prime}C^{\prime}}.
        \end{align}
        Then,
        \begin{multline}
            p(a, b,c \vert x, y,z)=\operatorname{Tr} \!\left[\Pi^{(x)}_{a} \otimes \Pi^{(y)}_{b} \otimes \Pi^{(z)}_{c}\otimes\ketbra*{000}{000}_{A^{\prime} B^{\prime}C^{\prime}}(\rho_{\tilde{A}\tilde{B}\tilde{C}A^{\prime} B^{\prime}C^{\prime}})\right]+\\ 
            \operatorname{Tr}\!\left[\Lambda^{(x)}_{a} \otimes \Lambda^{(y)}_{b}\otimes \Lambda^{(z)}_{c} \otimes\ketbra*{111}{111}_{A^{\prime} B^{\prime}C^{\prime}}(\rho_{\tilde{A}\tilde{B}\tilde{C}A^{\prime} B^{\prime}C^{\prime}})\right], 
        \end{multline}
        where it is understood that Alice is measuring $\sigma_{Z}$ on her system $A^{\prime}$, Bob is measuring $\sigma_{Z}$ on $B^{\prime}$, and Charlie is measuring $\sigma_{Z}$ on $C^{\prime}$, in addition to the other measurements on their systems $A$, $B$, and $C$. Now, consider the following quantum extension of $\rho_{A B C A^{\prime} B^{\prime}C^{\prime}}$,
        \begin{align}
            \rho_{\tilde{A}\tilde{B}\tilde{C}A^{\prime} B^{\prime}C^{\prime}}=\lambda \tau_{\tilde{A}\tilde{B}\tilde{C}E}\otimes\ketbra*{0000}{0000}_{A^{\prime}B^{\prime}C^{\prime}E^{\prime}}+ (1-\lambda)\gamma_{\tilde{A}\tilde{B}\tilde{C}E}\otimes\ketbra*{1111}{1111}_{A^{\prime} B^{\prime}C^{\prime}E^{\prime}}.
        \end{align}
        Furthermore, consider the following particular quantum extension of $\zeta_{ABCEXYZ}$:
        \begin{multline}
            \rho_{ABCXYZEE^{\prime}}= \sum_{a,b,c,x,y,z}  p(x,y,z)\{( \lambda)t(a,b,c\vert x,y,z) [a,b,c,x,y,z]\otimes\tau_{E}^{abcxyz}\otimes[0]_{E^{\prime}} \nonumber\\
            + (1 - \lambda)r(a,b,c\vert x,y,z) [a,b,c,x,y,z]\otimes\gamma_{E}^{abcxyz}\otimes[1]_{E^{\prime}}\}.
        \end{multline}
        Then following similar arguments given in the proof of Theorem~\ref{theorem:tri-conv}, we obtain
        \begin{align}
            N_{Q}(A;B;C)_{p} \leq \lambda N_{Q}(A;B;C)_{t} + (1 - \lambda) N_{Q}(A;B;C)_{r}.
        \end{align}
        This concludes the proof.
    \end{proof}

    \section{Monotonicity under Local Operations and Common Randomness}

\label{subsec:tri-mono-locr}

    Local Operations and Common Randomness (LOCR) is the set of free operations within the setup of conference key agreement. These free operations are chosen so that they are consistent with the prerequisites of the parity-CHSH game \cite{ribeiro2018fully}, which are similar to those of the CHSH game \cite{PhysRevLett.102.120401,PhysRevA.84.042112}. By common randomness, we mean that all parties have access to a common random variable and an instance that is made available to all parties before each round of the protocol. Using this common randomness, all parties can perform local operations and pre- and post-processing on their inputs and outputs. LOCR can be applied to an input distribution $p_{i}(a,b,c\vert x,y,z)$ to arrive at an output distribution $p_{f}(a_{f},b_{f},c_{f}\vert x_{f},y_{f},z_{f})$ as follows:
    \begin{multline}
        p_{f}(a_{f},b_{f},c_{f}\vert x_{f},y_{f},z_{f}) = \sum_{a,b,c,x,y,z} O^{(L)}(a_{f},b_{f},c_{f}\vert x_{f},y_{f},z_{f},a,b,c,x,y,z)\\
        p_{i}(a,b,c\vert x,y,z)I^{(L)}(x,y,z\vert x_{f},y_{f},z_{f}),
        \label{eq:input-output-rel-LOCR}
    \end{multline}
    where
    \begin{multline}\label{box1}
        O^{(L)}(a_{f},b_{f},c_{f}\vert x_{f},y_{f},z_{f},a,b,c,x,y,z) = \sum_{\lambda_{2}} p(\lambda_{2})O_{A}(a_{f}\vert a,x,x_{f},\lambda_{2})\\
        \times O_{B}(b_{f}\vert b,y,y_{f},\lambda_{2})O_{C}(c_{f}\vert c,z,z_{f},\lambda_{2}),
    \end{multline} 
    and 
    \begin{align}\label{box2}
        I^{(L)}(x,y,z\vert x_{f},y_{f},z_{f}) = \sum_{\lambda_{1}} p(\lambda_{1})I_{A}(x\vert x_{f},\lambda_{1})I_{B}(y\vert y_{f},\lambda_{1})I_{C}(z\vert z_{f},\lambda_{1}).
    \end{align}
    The bipartite case has been considered previously in \cite{PhysRevA.95.032118}.
    In the above equations, $O_{A}$, $O_{B}$, $O_{C}$, $I_{A}$, $I_{B}$, and $I_{C}$ are the pre-agreed local operations, and $\lambda_{1}$ and $\lambda_{2}$ represent the common randomness shared between the parties before and after obtaining the outputs from the initial correlation, respectively.
    
    \begin{theorem}[Monotonicity under LOCR]\label{theorem:tri-mono-locr}
        Let $p_{i}(a,b,c\vert x,y,z)$ be a no-signaling correlation, and let $p_{f}(a_{f},b_{f},c_{f}\vert x_{f},y_{f},z_{f})$ result from the action of local operations and common randomness on $p_{i}(a,b,c\vert x,y,z)$, as described in \eqref{eq:input-output-rel-LOCR}. Then,
        \begin{align}
            N(A_{i};B_{i};C_{i})_{p_{i}} \geq N(A_{i};B_{f};C_{f})_{p_{f}}.
            \label{eq:INL-monotone-ineq}
        \end{align}
    \end{theorem}
    \begin{proof}
        Consider the following respective no-signaling extensions of $p_{f}(a_{f},b_{f},c_{f}\vert x_{f},y_{f},z_{f})$ and $p_{i}(a,b,c\vert x,y,z)$:
        \begin{multline}\label{ext4}
            \zeta_{A_{f}B_{f}C_{f}EX_{f}Y_{f}Z_{f}} = \sum_{a_{f},b_{f},c_{f},x_{f},y_{f},z_{f}} q(x_{f},y_{f},z_{f})p_{f}(a_{f},b_{f},c_{f}\vert x_{f},y_{f},z_{f})\\ [a_{f}b_{f}c_{f}x_{f}y_{f}z_{f}]_{A_{f}B_{f}C_{f}EX_{f}Y_{f}Z_{f}}\otimes\zeta_{E}^{a_{f},b_{f},c_{f} x_{f},y_{f},z_{f}},
        \end{multline}
        and
        \begin{align}\label{ext5}
            \tau_{ABCEXYZ}= \sum_{a,b,c,x,y,z} q(x,y,z)p_{i}(a,b,c\vert x,y,z)[abcxyz]_{ABCEXYZ}\otimes\rho_{E}^{abcxyz}. 
        \end{align}
        Let us embed $p_{f}(a_{f},b_{f},c_{f}\vert x_{f},y_{f},z_{f})$ in the following quantum state:
        \begin{multline}
            \rho_{A_{f}B_{f}C_{f}X_{f}Y_{f}Z_{f}}= \sum_{a_{f},b_{f},c_{f},x_{f},y_{f},z_{f}} q(x_{f},y_{f},z_{f})\sum_{a,b,c,x,y,z}\sum_{\lambda_{2}}p(\lambda_{2})O_{A}(a_{f}\vert a,x,x_{f},\lambda_{2})\\
            O_{B}(b_{f}\vert b,y,y_{f},\lambda_{2})O_{C}(c_{f}\vert c,z,z_{f},\lambda_{2})p_{i}(a,b,c\vert x,y,z)\sum_{\lambda_{1}}p(\lambda_{1})I_{A}(x\vert x_{f},\lambda_{1})I_{B}(y\vert y_{f},\lambda_{1})\\ I_{B}(y\vert y_{f},\lambda_{1})I_{C}(z\vert z_{f},\lambda_{1})[a_{f}b_{f}c_{f}x_{f}y_{f}z_{f}]_{A_{f}B_{f}C_{f}X_{f}Y_{f}Z_{f}}.
        \end{multline}
        A particular no-signaling extension of this state is as follows:
        \begin{align}
            &\rho_{ABCA_{f}B_{f}C_{f}EXYZX_{f}Y_{f}Z_{f}\Lambda_{1}\Lambda_{2}}\nonumber\\
            &= \sum_{a_{f},b_{f},c_{f},x_{f},y_{f},z_{f}} q(x_{f},y_{f},z_{f})\sum_{a,b,c,x,y,z}\sum_{\lambda_{2}}p(\lambda_{2})O_{A}(a_{f}\vert a,x,x_{f},\lambda_{2})O_{B}(b_{f}\vert b,y,y_{f},\lambda_{2})\times\nonumber\\
            &\qquad \qquad O_{C}(c_{f}\vert c,z,z_{f},\lambda_{2})p_{i}(a,b,c\vert x,y,z)\sum_{\lambda_{1}}p(\lambda_{1})I_{A}(x\vert x_{f},\lambda_{1})I_{B}(y\vert y_{f},\lambda_{1})I_{B}(y\vert y_{f},\lambda_{1})\times\nonumber\\
            &\qquad \qquad I_{C}(z\vert z_{f},\lambda_{1}) [abcxyza_{f}b_{f}c_{f}x_{f}y_{f}z_{f}]_{ABCA_{f}B_{f}C_{f}XYZX_{f}Y_{f}Z_{f}}\otimes\rho_{E}^{abcxyz}\otimes[\lambda_{1}\lambda_{2}]_{\Lambda_{1}\Lambda_{2}}.
        \end{align}
        Now let us begin with the following inequality:
        \begin{align}
            \inf_{\text{ext. in }\eqref{ext4}} I(A_{f};B_{f};C_{f}\vert EX_{f}Y_{f}Z_{f})_{\zeta} \leq I(A_{f};B_{f};C_{f}\vert EX_{f}Y_{f}Z_{f}\Lambda_{1}\Lambda_{2})_{\rho}.
        \end{align}
        The above inequality holds for a   specific choice $\rho_{ABCA_{f}B_{f}C_{f}EXYZX_{f}Y_{f}Z_{f}\Lambda_{1}\Lambda_{2}}$ of a no-signaling extension  of $p_{f}(a_{f},b_{f},c_{f}\vert x_{f},y_{f},z_{f})$. Using data processing of conditional total correlation under local channels, we find that
        \begin{align}
            I(A_{f};B_{f};C_{f}\vert EX_{f}Y_{f}Z_{f}\Lambda_{1}\Lambda_{2})_{\rho} \leq I(AXX_{f}\Lambda_{2};BYY_{f}\Lambda_{2};CZZ_{f}\Lambda_{2}\vert EX_{f}Y_{f}Z_{f}\Lambda_{1}\Lambda_{2})_{\rho}.
        \end{align}
        Since $X_{f},Y_{f},Z_{f}$, and $\Lambda_{2}$ are classical copies of themselves, it follows that
        \begin{multline}
            I(AXX_{f}\Lambda_{2};BYY_{f}\Lambda_{2};CZZ_{f}\Lambda_{2}\vert EX_{f} Y_{f} Z_{f}\Lambda_{1}\Lambda_{2})_{\rho}\\
            = I(AX;BY;CZ\vert EX_{f}Y_{f}Z_{f}\Lambda_{1}\Lambda_{2})_{\rho}.
        \end{multline}
        Since none of $A$, $X$, $B$, $Y$, $C$,  and $Z$ depend on $\Lambda_{2}$, we conclude that
        \begin{align}
            I(AX;BY;CZ\vert EX_{f}Y_{f}Z_{f}\Lambda_{1}\Lambda_{2})_{\rho} = I(AX;BY;CZ\vert EX_{f}Y_{f}Z_{f}\Lambda_{1})_{\rho}.
        \end{align}
        Hence,
        \begin{align}
            \inf_{\text{ext. in }\eqref{ext4}} I(A_{f};B_{f};C_{f}\vert EX_{f}Y_{f}Z_{f})_{\zeta} \leq I(AX;BY;CZ\vert EX_{f}Y_{f}Z_{f}\Lambda_{1})_{\rho} .
        \end{align}
        Using \eqref{eqnA}, we find that
        \begin{multline}
            I(AX;BY;CZ\vert EX_{f}Y_{f}Z_{f}\Lambda_{1})_{\rho}= I(A;B;C\vert EX_{f}Y_{f}Z_{f}\Lambda_{1}XYZ)_{\rho}\\
            + I(X;Y;Z\vert YEX_{f}Y_{f}Z_{f}\Lambda_{1})_{\rho} + I(YZ;A\vert XEX_{f}Y_{f}Z_{f}\Lambda_{1})_{\rho} \\
            + I(XZ;B\vert YEX_{f}Y_{f}Z_{f}\Lambda_{1})_{\rho} + I(XY;C\vert ZEX_{f}Y_{f}Z_{f}\Lambda_{1})_{\rho}.
        \end{multline}
        The information-theoretic quantities $I(YZ;A\vert XEX_{f}Y_{f}Z_{f}\Lambda_{1})_{\rho}$, $I(XZ;B\vert YEX_{f}Y_{f}Z_{f}\Lambda_{1})_{\rho}$, and $I(XY;C\vert ZEX_{f}Y_{f}Z_{f}\Lambda_{1})_{\rho}$ are equal to zero due to the no-signaling constraints elucidated in \eqref{nosig2} and the structure in \eqref{box2} of the local box $I_L$. The information-theoretic quantity $I(X;Y;Z\vert YEX_{f}Y_{f}Z_{f}\Lambda_{1})_{\rho}$ is equal to zero due to \eqref{box2}. The structure of $\rho$  implies that all the terms are equal to zero, except for the first term. So,
        \begin{align}
            \inf_{\text{ext. in }\eqref{ext4}} I(A_{i};B_{f};C_{f}\vert EX_{f}Y_{f}Z_{f})_{\zeta} & \leq I(A;B;C\vert XYZEX_{f}Y_{f}Z_{f}\Lambda_{1})_{\rho}\\
            & = I(A;B;C\vert XYZE)_{\tau},
        \end{align}
        where the equality is a consequence of the structure of  $\rho_{ABCEXYZX_{f}Y_{f}Z_{f}\Lambda_{1}}$.
        Since $\tau$ is an arbitrary no-signaling extension of $p_i$, we conclude that
        \begin{align}
            \inf_{\text{ext. in }\eqref{ext4}} I(A_{i};B_{f};C_{f}\vert EX_{f}Y_{f}Z_{f})_{\zeta} \leq \inf_{\text{ext. in }\eqref{ext5}} I(A;B;C\vert XYZE)_{\tau}.
        \end{align}
        By optimizing over arbitrary input probability distributions, we conclude that
        \begin{align}
            \sup_{q}\inf_{\text{ext. in }\eqref{ext4}} I(A_{f};B_{f};C_{f}\vert EX_{f}Y_{f}Z_{f})_{p_{f}} &\leq \sup_{q}\inf_{\text{ext. in }\eqref{ext5}} I(A;B;C\vert XYZE)_{p_{i}},
        \end{align}
        which is the desired inequality in \eqref{eq:INL-monotone-ineq}.
    \end{proof}

    \begin{theorem}[Monotonicity under LOCR of QTINL]\label{theorem:tri-mono-locr-q}
        Let $p_{i}(a,b,c\vert x,y,z)$ be a quantum correlation, and let $p_{f}(a_{f},b_{f},c_{f}\vert x_{f},y_{f},z_{f})$ result from the action of local operations and common randomness on $p_{i}(a,b,c\vert x,y,z)$, as described in \eqref{eq:input-output-rel-LOCR}. Then
        \begin{align}
            N_{Q}(A_{i};B_{i};C_{i})_{p_{i}} \geq N_{Q}(A_{i};B_{f};C_{f})_{p_{f}}.
            \label{eq:q-intr-nl-monotoone}
        \end{align}
    \end{theorem}
    \begin{proof}
        First, let us embed $p_{f}(a_{f},b_{f},c_{f}\vert x_{f},y_{f},z_{f})$ in a quantum state:
        \begin{multline}
            \zeta_{A_{f}B_{f}C_{f}X_{f}Y_{f}Z_{f}} = \\ \sum_{a_{f},b_{f},c_{f},x_{f},y_{f},z_{f}} q(x_{f},y_{f},z_{f})p_{f}(a_{f},b_{f},c_{f}\vert x_{f},y_{f},z_{f}) [a_{f}b_{f}c_{f}x_{f}y_{f}z_{f}]_{A_{f}B_{f}C_{f} X_{f}Y_{f}Z_{f}},
        \end{multline}
    where $q(x_{f},y_{f},z_{f})$ is an arbitrary probability distribution for $x_{f}$, $y_{f}$, and $z_{f}$. The set~$\mathcal{Q}$ of quantum correlations  is closed under the action of local operations and common randomness, implying that $p_{f}(a_{f},b_{f},c_{f}\vert x_{f},y_{f},z_{f}) \in\mathcal{Q}$. Since $p_{f}(a_{f},b_{f},c_{f}\vert x_{f},y_{f},z_{f})$ is also a quantum correlation, we know that there exists an underlying state $\zeta_{\tilde{A}_{f}\tilde{B}_{f}\tilde{C}_{f}}$and POVMs $\{\Pi^{(x_{f})}_{a_{f}}\}_{a_{f}}$, $\{\Pi^{(y_{f})}_{b_{f}}\}_{b_{f}}$, and $\{\Pi^{(z_{f})}_{c_{f}}\}_{c_{f}}$ such that
    \begin{align}
        p_{f}(a_{f},b_{f},c_{f}\vert x_{f},y_{f},z_{f})= \operatorname{Tr}\!\left[\left(\Pi^{(x_{f})}_{a_{f}}\otimes\Pi^{(y_{f})}_{b_{f}}\otimes\Pi^{(z_{f})}_{c_{f}}\right)\zeta_{\tilde{A}_{f}\tilde{B}_{f}\tilde{C}_{f}}\right].
    \end{align}
    An arbitrary quantum extension of the state $\zeta_{A_{f}B_{f}C_{f}X_{f}Y_{f}Z_{f}}$ is given by
    \begin{multline}\label{eqn4-q}
        \zeta_{A_{f}B_{f}C_{f}EX_{f}Y_{f}Z_{f}} = \sum_{a_{f},b_{f},c_{f},x_{f},y_{f},z_{f}} q(x_{f},y_{f},z_{f})p_{f}(a_{f},b_{f},c_{f}\vert x_{f},y_{f},z_{f})\\
        [a_{f}b_{f}c_{f}x_{f}y_{f}z_{f}]_{A_{f}B_{f}C_{f}EX_{f}Y_{f}Z_{f}}\otimes\zeta_{E}^{a_{f},b_{f},c_{f} x_{f},y_{f},z_{f}},
    \end{multline}
    where
    \begin{align}
        \zeta_{E}^{a_{f},b_{f},c_{f} x_{f},y_{f},z_{f}}=\frac{1}{p_{f}(a_{f},b_{f},c_{f}\vert x_{f},y_{f},z_{f})}\operatorname{Tr}\!\left[\left(\Pi^{(x_{f})}_{a_{f}}\otimes\Pi^{(y_{f})}_{b_{f}}\otimes\Pi^{(z_{f})}_{c_{f}}\otimes\mathbb{I}\right)\zeta_{\tilde{A}_{f}\tilde{B}_{f}\tilde{C}_{f}E}\right],
    \end{align}
    and $\zeta_{A_{f}B_{f}C_{f}E}$ is an quantum extension of $\zeta_{A_{f}B_{f}C_{f}}$. Now, we know that
    \begin{multline}
            p_{f}(a_{f},b_{f},c_{f}\vert x_{f},y_{f},z_{f}) = \sum_{a,b,c,x,y,z} O^{(L)}(a_{f},b_{f},c_{f}\vert x_{f},y_{f},z_{f},a,b,c,x,y,z)\\
            p_{i}(a,b,c\vert x,y,z)I^{(L)}(x,y,z\vert x_{f},y_{f},z_{f}),
        \end{multline}
    as well as the facts that $I^{(L)}(x,y,z\vert x_{f},y_{f},z_{f})$ and $O^{(L)}(a_{f},b_{f},c_{f}\vert x_{f},y_{f},z_{f},a,b,c,x,y,z)$ are local correlations. Therefore, there exist separable states $\zeta_{XYZ}$ and $\rho_{\tilde{A}_{f}\tilde{B}_{f}\tilde{C}_{f}}$, along with POVMs that result in the correlations $I^{(L)}$ and $O^{(L)}$. That is,
    \begin{align}
        I^{(L)}(x,y,z\vert x_{f},y_{f},z_{f})= \operatorname{Tr}\!\left[\left(\Pi_{x}^{(x_{f})}\otimes\Pi_{y}^{(y_{f})}\otimes\Pi_{z}^{(z_{f})}\right)\zeta_{XYZ}\right],
    \end{align}
    and 
    \begin{align}
        O^{(L)}(a_{f},b_{f},c_{f}\vert x_{f},y_{f},z_{f},a,b,c,x,y,z)= \operatorname{Tr}\!\left[\left(\Pi_{a_{f}}^{(x_{f},a,x)}\otimes\Pi_{b_{f}}^{(y_{f},b,y)}\otimes\Pi_{c_{f}}^{(z_{f},c,z)}\right)\rho_{\tilde{A}_{f}\tilde{B}_{f}\tilde{C}_{f}}\right].
    \end{align}
    Furthermore, we know that the correlation $p_{i}(a,b,c\vert x,y,z)$ is a quantum correlation. Thus, there exists an underlying state $\zeta_{\tilde{A}\tilde{B}\tilde{C}}$ and POVMs $\{\Pi^{(x)}_{a}\}_a$, $\{\Pi^{(y)}_{b}\}_b$, and $\{\Pi^{(z)}_{c}\}_b$ such that
    \begin{multline}
        p(a_{f}, b_{f}, c_{f} \vert x_{f}, y_{f}, z_{f}) =\\
        \sum_{a,b,c,x,y,z} \operatorname{Tr}\!\Big[\left(\Pi_{a_{f}}^{(x_{f},a,x)}\otimes\Pi_{b_{f}}^{(y_{f},b,y)}\otimes\Pi_{c_{f}}^{(z_{f},c,z)}\otimes\Pi_{x}^{(x_{f})}\otimes\Pi_{y}^{(y_{f})}\otimes\Pi_{z}^{(z_{f})}\otimes\Pi^{(x)}_{a}\otimes\Pi^{(y)}_{b}\otimes\Pi^{(z)}_{c}\right)\\\left(\rho_{\tilde{A}_{f}\tilde{B}_{f}\tilde{C}_{f}}\otimes\zeta_{XYZ}\otimes \zeta_{\tilde{A}\tilde{B}\tilde{C}}\right)
        \Big].
    \end{multline}
    Since $\zeta_{XYZ}$ is a separable state, we can write it as $\zeta_{XYZ}=\sum_{\lambda_{1}} p\left(\lambda_{1}\right) \zeta_{X}^{\lambda_{1}} \otimes \zeta_{Y}^{\lambda_{1}} \otimes \zeta_{Z}^{\lambda_{1}}.$ Let  $\zeta_{XYZ\Lambda_{1}}=\sum_{\lambda_{1}} p\left(\lambda_{1}\right) \zeta_{X}^{\lambda_{1}} \otimes \zeta_{Y}^{\lambda_{1}} \otimes \zeta_{Z}^{\lambda_{1}}\otimes[\lambda_{1}]_{\Lambda_{1}}$ be a particular quantum extension of $\zeta_{XYZ}$. Similarly, let $\rho_{\tilde{A}_{f}\tilde{B}_{f}\tilde{C}_{f}\Lambda_{2}}$ be a quantum extension of $\rho_{\tilde{A}_{f}\tilde{B}_{f}\tilde{C}_{f}}$ and $\zeta_{\tilde{A}\tilde{B}\tilde{C}E}$ an extension of~$\zeta_{\tilde{A}\tilde{B}\tilde{C}}$.
    A particular quantum extension of the state in \eqref{eqn4-q} is given by
    \begin{multline}
        \rho_{A_{f}B_{f}C_{f}EX_{f}Y_{f}Z_{f}\Lambda_{1}\Lambda_{2}}= \sum_{a_{f},b_{f},c_{f},x_{f},y_{f},z_{f}} p(x_{f},y_{f},z_{f}) q_{f}(a_{f},b_{f},c_{f}\vert x_{f},y_{f},z_{f})\times \\ [a_{f}b_{f}c_{f}x_{f}y_{f}z_{f}]_{A_{f}B_{f}C_{f}X_{f}Y_{f}Z_{f}}\otimes\rho_{E}^{abcxyz}\otimes[\lambda_{1}\lambda_{2}]_{\Lambda_{1}\Lambda_{2}},
    \end{multline}
    where
    \begin{align}
        \rho_{E}^{a,b,c, x,y,z}=\frac{1}{p(a,b,c,\vert x,y,z)}\operatorname{Tr}\!\left[\left(\Pi^{(x)}_{a}\otimes\Pi^{(y)}_{b}\otimes\Pi^{(z)}_{c}\otimes\mathbb{I}\right)\zeta_{\tilde{A}\tilde{B}\tilde{C}E}\right],
    \end{align}
    which then gives
    \begin{align}
        &\rho_{ABCA_{f}B_{f}C_{f}EXYZX_{f}Y_{f}Z_{f}\Lambda_{1}\Lambda_{2}}\nonumber\\
        &= \sum_{a_{f},b_{f},c_{f},x_{f},y_{f},z_{f}} q(x_{f},y_{f},z_{f})\sum_{a,b,c,x,y,z}\sum_{\lambda_{2}}p(\lambda_{2})O_{A}(a_{f}\vert a,x,x_{f},\lambda_{2})O_{B}(b_{f}\vert b,y,y_{f},\lambda_{2})\times\nonumber\\
        &\qquad \qquad O_{C}(c_{f}\vert c,z,z_{f},\lambda_{2})p_{i}(a,b,c\vert x,y,z)\sum_{\lambda_{1}}p(\lambda_{1})I_{A}(x\vert x_{f},\lambda_{1})I_{B}(y\vert y_{f},\lambda_{1})I_{B}(y\vert y_{f},\lambda_{1})\times\nonumber\\
        &\qquad \qquad I_{C}(z\vert z_{f},\lambda_{1}) [abcxyza_{f}b_{f}c_{f}x_{f}y_{f}z_{f}]_{ABCA_{f}B_{f}C_{f}XYZX_{f}Y_{f}Z_{f}}\otimes\rho_{E}^{abcxyz}\otimes[\lambda_{1}\lambda_{2}]_{\Lambda_{1}\Lambda_{2}}.
    \end{align}
    Then, following arguments similar to that given in Theorem~\ref{theorem:tri-mono-locr}, we obtain the desired inequality in~\eqref{eq:q-intr-nl-monotoone}.
    \end{proof}

    \section{Local Hidden-Variable Models}\label{subsec:local-cor}
    
    In this appendix, we show that tripartite intrinsic non-locality and quantum tripartite intrinsic non-locality vanish for tripartite correlations that admit a local hidden-variable model. A tripatite correlation admits a local hidden-variable model if it is of the following form \cite{brunner2014bell}: 
    \begin{equation}
        p(a,b,c\vert x,y,z)=\sum_{\Lambda}p_{\Lambda}(\lambda)p(a\vert x,\lambda)p(b\vert y,\lambda)p(c\vert z,\lambda),
    \end{equation}
    where $\lambda$ is a local hidden variable. If a distribution admits such a model, then the model can be reformulated so that all the factor distributions $p(a|x,\lambda)$, $p(b|y,\lambda)$, and $p(c|z,\lambda)$ are deterministic with probabilities equal to either zero or one. In this case, using the classical information $\lambda$ and the input settings of $x$, $y$, and $z$, an eavesdropper can deduce the outcomes $a$, $b$, and $c$ with certainty. Hence, tripartite intrinsic non-locality and quantum tripartite intrinsic non-locality should vanish for local tripartite correlations.
    
    \begin{theorem}[TINL \& QTINL for local correlations]\label{theorem:local-cor}
        Tripartite intrinsic non-locality and quantum tripartite intrinsic non-locality vanish for every distribution $p(a,b,c\vert x,y,z)$ having a local hidden-variable model, i.e., $N(A;B;C)_{p} =0$ and $N_{Q}(A;B;C)_{p} =0$.
    \end{theorem}
    \begin{proof}
    Consider the following no-signaling extension of $p(a,b,c\vert x,y,z)$:
        \begin{align}\label{ext:local-gen}
            \zeta_{ABCEXYZ} & = \sum_{a,b,c,x,y,z} q(x,y,z)p(a,b,c\vert x,y,z)[abcxyz]_{ABCXYZ}\otimes\rho_{E}^{abcxyz}.
        \end{align}
        A particular no-signaling extension of $p(a,b,c\vert x,y,z)$ is
        \begin{align}
            \rho_{ABCEXYZ} & = \sum_{a,b,c,x,y,z,\Lambda} p_{\Lambda}(\lambda)q(x,y,z)p(a\vert x,\lambda)p(b\vert y,\lambda)p(c\vert z,\lambda)[abcxyz]_{ABCXYZ}\otimes[\lambda]_{E}\nonumber\\
            & = \sum_{\Lambda} p_{\Lambda}(\lambda)\rho^\lambda_{ABCXYZ}\otimes[\lambda]_{E}
        \end{align}
        where 
        \begin{align}\label{ext:local-parti}
        \rho^\lambda_{ABCXYZ}= \sum_{a,b,c,x,y,z} q(x,y,z)p(a\vert x,\lambda)p(b\vert y,\lambda)p(c\vert z,\lambda)[abcxyz]_{ABCXYZ}.
        \end{align} 
        Then, it follows that
        \begin{align}
            \inf_{\text{ext. in }\eqref{ext:local-gen}} I(A;B;C \vert EXYZ)_{\zeta}\leq I(A;B;C \vert EXYZ)_{\rho}=\sum_{\Lambda} p_{\Lambda}(\lambda)I(A;B;C \vert XYZ)_{\rho_\lambda}
        \end{align}
        From inspection of \eqref{ext:local-parti}, we conclude that $I(A;B;C \vert XYZ)_{\rho_\lambda}=0$. Therefore, we obtain the first desired claim: $N(A;B;C)_{p} =0$. One can see that $N_{Q}(A;B;C)_{p} =0$ by considering the appropriate quantum extensions.
    \end{proof}
\end{document}